\RequirePackage{amsmath,amssymb,amsfonts}
\documentclass[runningheads]{llncs}
\DeclareMathAlphabet{\mathpzc}{OT1}{pzc}{m}{it}

\newcommand{\tuple}[1]{( #1 )}
\newcommand{\bigtuple}[1]{\big( #1 \big)}

\newcommand{\abs}[1]{\left| #1 \right|} 

\newcommand{\NN}{\mathbb{N}}    
   
\newcommand{\EE}{\mathbb{E}} 

\newcommand{\A}{\mathcal{A}}
\newcommand{\B}{\mathcal{B}}

\newcommand{\samp}{\leftarrow}
\newcommand{\lruns}{\leftarrow}

\newcommand{\bin}{ \{ 0,1 \} }

\newcommand{\secpar}{\kappa}

\newcommand{\msg}{\mathsf{M}}

\newcommand{\seed}{\mathsf{seed}}

\newcommand{\Adv}{\mathsf{Adv}}

\newcommand{\PKE}{\mathsf{PKE}}
\newcommand{\KEM}{\mathsf{KEM}}
\newcommand{\Setup}{\mathsf{Setup}}
\newcommand{\Gen}{\mathsf{Gen}}
\newcommand{\Enc}{\mathsf{Enc}}
\newcommand{\Encaps}{\mathsf{Encaps}}
\newcommand{\Decaps}{\mathsf{Decaps}}
\newcommand{\EncA}{\Enc^{\sf i}}
\newcommand{\EncB}{\Enc^{\sf d}}

\newcommand{\Dec}{\mathsf{Dec}}
\newcommand{\Ext}{\mathsf{Ext}}

\newcommand{\KeySpace}{\mathcal{K}}
\newcommand{\MsgSpace}{\mathcal{M}}
\newcommand{\RandSpace}{\mathcal{R}}
\newcommand{\CtSpace}{\mathcal{C}}
\newcommand{\CtSpaceSin}{\CtSpace_{\sf single}}

\newcommand{\MR}{\mathsf{m}}
\newcommand{\mPKE}{\MR\PKE}

\newcommand{\mKEM}{\MR\KEM}
\newcommand{\mSetup}{\MR\Setup}
\newcommand{\mSetupp}{\MR\Setup^{\sf p}}
\newcommand{\mGen}{\MR\Gen}
\newcommand{\mGenp}{\MR\Gen^{\sf p}}
\newcommand{\mEnc}{\MR\Enc}

\newcommand{\mEncA}{\MR\EncA}
\newcommand{\mEncB}{\MR\EncB}

\newcommand{\mDec}{\MR\Dec}

\newcommand{\mExt}{\MR\Ext}
\newcommand{\mExtp}{\MR\Ext^{\sf p}}
\newcommand{\mEncaps}{\MR\Encaps}

\newcommand{\mDecaps}{\MR\Decaps}
\newcommand{\OD}{\mathcal{D}} 

\newcommand{\rand}{\mathsf{r}}

\newcommand{\randnpk}{\rand_0}

\newcommand{\ct}{\mathsf{ct}}
\newcommand{\vct}{\vec{\sf ct}}

\newcommand{\hct}{\widehat{\sf ct}}
\newcommand{\ctnpk}{\ct_0}
\newcommand{\pp}{\mathsf{pp}}
\newcommand{\pk}{\mathsf{pk}}
\newcommand{\sk}{\mathsf{sk}}

\newcommand{\skp}{\mathsf{sk}^{\sf p}}

\newcommand{\key}{\mathsf{K}}

\newcommand{\state}{\mathsf{state}}

\newcommand{\IND}{\mathsf{IND}\text{-}}

\newcommand{\CPA}{\mathsf{CPA}}
\newcommand{\CCA}{\mathsf{CCA}}
\newcommand{\INDCPA}{\IND\CPA}
\newcommand{\INDCCA}{\IND\CCA}

\newcommand{\BIND}{{\B_{\sf IND}}}

\newcommand{\OG}{\mathsf{G}}

\newcommand{\OH}{\mathsf{H}}

\newcommand{\qD}{q_\OD}
\newcommand{\qG}{q_\OG}
\newcommand{\qH}{q_\OH}

\makeatletter
\newcommand{\makeAlph}[1]{\@alph{#1}}
\makeatother

\usepackage[T1]{fontenc}
\usepackage{graphicx}
\usepackage{multirow}%
\usepackage{mathtools}

\usepackage{mathrsfs}%
\usepackage[title]{appendix}%
\usepackage{xcolor}%
\usepackage{textcomp}%
\usepackage{manyfoot}%
\usepackage{booktabs}%
\usepackage{algorithm}%
\usepackage{algpseudocode}%
\usepackage{listings}%
\usepackage{multirow}

\usepackage{eurosym}
\usepackage{float}
\usepackage[tight,footnotesize]{subfigure}
\usepackage{rotating}
\usepackage{enumerate}
\usepackage{color,soul}
\usepackage[colorinlistoftodos]{todonotes}
\usepackage{lipsum}
\usepackage{cases}
\usepackage{comment}
\usepackage{resizegather}
\usepackage{xurl}
\usepackage{tablefootnote}
\usepackage{array}
\newcolumntype{C}[1]{>{\centering\arraybackslash$}p{#1}<{$}}
\usepackage{makecell}

\usepackage{threeparttable}

\usepackage{xpatch}
\usepackage{nicematrix}

\usepackage{tikz}
\usetikzlibrary{tikzmark,calc,fit}
\usepackage{tikzpagenodes} 
\usepackage{booktabs}

\newcommand{\bs}{\boldsymbol{\mathrm{s}}}
\newcommand{\br}{\boldsymbol{\mathrm{r}}}
\newcommand{\be}{\boldsymbol{\mathrm{e}}}

\newcommand{\bu}{\boldsymbol{\mathrm{u}}}
\newcommand{\bt}{\boldsymbol{\mathrm{t}}}
\newcommand{\bA}{\boldsymbol{\mathrm{A}}}
\newcommand{\buone}{\bA^{T}\br+\be_{1}}

\newcommand{\coloneqq}{:=}
\usepackage{lscape}

\DeclareMathOperator{\diag}{diag}
 \DeclareMathOperator{\E}{E}

\DeclareMathSizes{10}{9}{7}{5}

\newlength{\myeqskip}  

\expandafter\def\expandafter\normalsize\expandafter{%
    \normalsize%
    \setlength\abovedisplayskip{\myeqskip}%
    \setlength\belowdisplayskip{\myeqskip}%
    \setlength\abovedisplayshortskip{\myeqskip}%
    \setlength\belowdisplayshortskip{\myeqskip}%
}

\begin{document}
\emergencystretch 3em
%


\title{Compact Lattice-Coded (Multi-Recipient) Kyber without CLT Independence Assumption}

\author{{Shuiyin} {Liu}\inst{1}\orcidID{0000-0002-3762-8550} \and
{Amin} {Sakzad}\inst{2}\orcidID{0000-0003-4569-3384} }

\authorrunning{S. Liu and A. Sakzad}
%
\institute{ {Holmes
Institute}, {Melbourne, VIC 3000},  {Australia}
\\ \email{SLiu@holmes.edu.au}\\
\and
{Monash University}, {Melbourne, VIC 3800}, {Australia}\\
\email{Amin.Sakzad@monash.edu}}

\maketitle              
\begin{abstract}
This work presents a joint design of encoding and encryption procedures for public key encryptions (PKEs) and key encapsulation mechanism (KEMs) such as Kyber, without relying on the assumption of independent decoding noise components, achieving reductions in both communication overhead (CER) and decryption failure rate (DFR). Our design features two techniques: ciphertext packing and lattice packing. First, we extend the Peikert-Vaikuntanathan-Waters (PVW) method to Kyber: $\ell$ plaintexts are packed into a single ciphertext. This scheme is referred to as P$_\ell$-Kyber. We prove that the P$_\ell$-Kyber is IND-CCA secure under the M-LWE hardness assumption. We show that the decryption decoding noise entries across the $\ell$ plaintexts (also known as layers) are mutually independent. Second, we propose a cross-layer lattice encoding scheme for the P$_\ell$-Kyber, where every $\ell$ cross-layer information symbols are encoded to a lattice point. This way we obtain a \emph{coded} P$_\ell$-Kyber, where the decoding noise entries for each lattice point are mutually independent. Therefore, the DFR analysis does not require the assumption of independence among the decryption decoding noise entries. Both DFR and CER are greatly decreased thanks to ciphertext packing and lattice packing. We demonstrate that with $\ell=24$ and Leech lattice encoder, the proposed coded P$_\ell$-KYBER1024 achieves DFR $<2^{-281}$ and CER $ = 4.6$, i.e., a decrease of CER by $90\%$ compared to KYBER1024. Additionally, for a fixed plaintext size matching that of standard Kyber ($256$ bits), we introduce a truncated variant of P$_\ell$-Kyber that deterministically removes ciphertext components carrying surplus information bits. Using $\ell=8$ and E8 lattice encoder, we show that the proposed truncated coded P$_\ell$-KYBER1024 achieves a $10.2\%$ reduction in CER and improves DFR by a factor of $2^{30}$ relative to KYBER1024. Finally, we demonstrate that constructing a multi-recipient PKE and a multi-recipient KEM (mKEM) using the proposed truncated coded P$_\ell$-KYBER1024 results in a $20\%$ reduction in bandwidth consumption compared to the existing schemes.

\end{abstract}
\vspace{-1mm} 
\section{Introduction}
\vspace{-1mm} 
In August 2024, National Institute of Standards and Technology (NIST) has published the final post-quantum cryptography standards for digital-signature, encryption, and key-encapsulation mechanisms (KEM). CRYSTALS-Kyber is the only post-quantum KEM standardised by NIST \cite{NISTpqcfinal2024}. In February 2024, Apple has announced that its iMessage is going to use Kyber \cite{ApplePQ32024}, where sender devices generate post-quantum encryption keys using the receiver’s public keys. Kyber is a lattice-based cryptographic algorithm built upon the module-learning with errors (M-LWE) problem. Unlike traditional KEMs like Elliptic-curve Diffie–Hellman (ECDH), the Kyber algorithm results in much larger ciphertext size (e.g., up to $49$ times larger), which necessitates more storage, increased memory usage, and greater demand for network bandwidth. Later, Facebook reported an increase of about $40\%$ in CPU cycles after implementing Kyber, compared with its current ECDH \cite{Meta2024}. For the upcoming large-scale deployment phase, it is crucial to enhance Kyber in order to reduce its storage, memory usage, and communication bandwidth. In early 2025, NIST announced adaptation of HQC as another KEM built upon code-based cryptographic assumptions.

Kyber’s encryption and decryption processes can be viewed as a noisy communication channel with binary Pulse Amplitude Modulation (2-PAM) \cite{LWEchannel2022}. Recent coding schemes aim to encode more bits, reducing the ciphertext expansion rate (CER). For instance, \cite{liu2023lattice} replaces 2-PAM with a Leech lattice constellation, achieving a $32.6\%$ CER reduction. \cite{LWEchannel2022} proposes a 5-PAM Q-ary BCH encoding, cutting CER by $45.6\%$. To the best of our knowledge, \cite{liu2024SCKyber} achieves the best results so far. That is a $54\%$ CER reduction by transforming Kyber’s processes into a Gaussian channel, encoding a $638$-bit secret in a single ciphertext. These methods rely on an independence assumption for decryption noise, using the central limit theorem to model noise as Gaussian. While valid for higher dimensions, this assumption may underestimate the decryption failure rate (DFR), as the noise entries are actually dependent.

Kyber uses a lossy compression (quantization) function to reduce ciphertext size, increasing decoding noise. Efforts to reduce quantization noise focus on minimizing channel noise. \cite{NewHope2016} applied $\mathsf{D_4}$-lattice quantization to Ring-LWE (R-LWE), and \cite{MLWEE82021} extended this to M-LWE with $\mathsf{E_8}$-lattice quantization. \cite{shuiyin2024} proposed a lattice quantization framework for Kyber, reducing the CER by $36.47\%$ and DFR by a factor of $2^{99}$. \cite{liu2024SCKyber} showed that Lloyd-Max quantization minimizes mean squared error (MMSE) for M-LWE samples. However, the DFR analysis in \cite{shuiyin2024}\cite{liu2024SCKyber} still assumes independence in decoding noise entries.

The independence assumption mentioned above is closely related to the Central Limit Theory (CLT): the dominant decoding noise components are assumed to be i.i.d. Gaussian random variables. The independent/CLT assumption also appears in the wider literature on (R/M-) LWE based cryptosystem and fully homomorphic encryption (FHE). In R-LWE, the independent assumption was applied to the decoding noise entries of NewHope-Simple \cite{NewhopeECC2018}. In LWE, the Gaussian approximation was applied to the decoding noise entries of FrodoKEM \cite{FrodoCong2022}. The CLT assumption was also used in the homomorphic encryption (HE) schemes \cite{TFHE2018}\cite{CLTRLWE2024}. An open question is how to create a (R/M-) LWE based \emph{coded} cryptosystem that does not depend on the independence assumption on the decoding noise entries.

Apart from the coding approach, an alternative way of reducing the CER is ciphertext packing \cite{PVW2008} (also well known as PVW packing), which was originally proposed for LWE. The authors show that a single LWE ciphertext vector (i.e., the first part of the ciphertext) can be securely reused to encrypt multiple ciphertexts 
(i.e., $1$-bit messages)  by employing various secret-key vectors. The resulting cryptosystem is much more efficient than Regev’s scheme \cite{Regev05}, since the CER can be made as small as a constant by packing many plaintexts. This method has been used to construct LWE based HE schemes \cite{homomorphiLWE2011}\cite{BGH2013} and KEM scheme \cite{FrodoKEM2021}. The downside of ciphertext packing is the increased DFR, due to a union bound probability of decryption failure for each plaintext. To the best of the authors' knowledge, the ciphertext packing has not been introduced to M-LWE based KEMs like Kyber. It would be interesting to see how the ciphertext packing method affects the DFR analysis, CER, and security level of Kyber.

Although Kyber is originally specified to encapsulate a $256$-bit secret, many key encapsulation mechanisms—particularly those used in contexts involving Forward Secrecy (FS) and Key Rotation—frequently exchange raw shared secrets exceeding $256$ bits. For instance, TLS 1.3 employs P-$384$ for FS, yielding a $384$-bit raw secret prior to key derivation \cite{rfc8446}. Apple’s PQ3 protocol employs a hybrid key exchange scheme that combines ECC and Kyber, producing a concatenated raw shared secret of approximately 512 bits prior to input to the key derivation function (KDF) \cite{ApplePQ32024}. From a multi-key derivation perspective, many protocols require derivation of several symmetric keys (e.g., AES key, HMAC key, IV) from a single exchange. Even when only a $256$-bit AES key is ultimately used, a larger raw secret enhances KDF resilience and entropy distribution. From a theoretical perspective, \cite{LWEchannel2022} demonstrated that encoding across four standard Kyber ciphertexts can yield a 2214-bit secret. More recently, \cite{liu2024SCKyber} showed that it is feasible to embed two AES keys within a single standard Kyber ciphertext. These findings suggest the potential for constructing more compact M-LWE-based KEMs than the current Kyber scheme.  Accordingly, we argue that extending Kyber to support the encapsulation of larger secrets is both practically meaningful and theoretically significant for modern cryptographic applications.

The standard single-recipient KEM can be generalized to support multiple recipients (mKEM) by accepting multiple public keys as input \cite{MR2005}. This is useful in scenarios where the same session key $\mathbf{m}$ needs to be securely shared with a group of recipients. A key advantage of mKEM is its reduced ciphertext overhead compared to the naive approach of performing individual encryptions for each recipient. In \cite{MR_PKE2020}, Kyber-based mPKE and mKEM schemes were introduced, wherein the session key $\mathbf{m}$ is encapsulated using multiple public keys and combined into a single ciphertext. This approach achieves a significant reduction in ciphertext size relative to the naive method. Moreover, any further reduction in the ciphertext size of the underlying Kyber scheme directly translates to a corresponding reduction in the ciphertext size of the Kyber-based mKEM.

Our main contribution is to develop a M-LWE based KEM with a very low CER (e.g., CER $<5$), where its DFR analysis does not depend on the assumption of independence among the decryption decoding noise entries. Our design leverages two techniques: ciphertext packing and lattice packing. The former greatly reduces the CER, while the latter is effective at decreasing the DFR (i.e., compensating the drawback of ciphertext packing). Below we summarize the means that we achieve this: 
\begin{itemize}
\item  We first propose a packed version of Kyber: $\ell$
plaintexts are packed into a single ciphertext. This scheme is referred to as P$_\ell$-Kyber. We prove that the P$_\ell$-Kyber is IND-CCA secure under the M-LWE hardness assumption. We show that the decryption decoding noise entries across the $\ell$ plaintexts (also known as layers) are mutually independent. We also propose a cross-layer lattice encoding scheme for the P$_\ell$-Kyber, where every $\ell$ cross-layer information symbols are encoded to a lattice point. This way we obtain a \emph{coded} P$_\ell$-Kyber, which takes the advantages of both ciphertext packing and lattice packing. An upper bound on the DFR of {coded} P$_\ell$-Kyber is derived, which can be verified numerically. We demonstrate that with $\ell=24$ and Leech lattice encoder, the proposed coded P$_\ell$-KYBER1024 achieves DFR $\leq 2^{-281}$ and CER $=4.6$ (see Table~\ref{KyberP_LC}). 
\item Secondly, for a fixed plaintext size equivalent to that of standard Kyber ($256$ bits), we propose a truncated variant of P$_\ell$-Kyber that deterministically eliminates ciphertext components conveying redundant information bits. Employing $\ell = 8$ in conjunction with the \text{E8} lattice encoder, the proposed truncated coded P$_\ell$-KYBER1024 achieves a $10.2\%$ reduction in the CER and yields a DFR improvement by a factor of $2^{30}$ compared to KYBER1024. Additionally, we finally demonstrate that implementing a multi-recipient KEM (mKEM) based on the proposed truncated coded P$_\ell$-KYBER1024 achieves a $20\%$ reduction in bandwidth usage compared to existing mKEM schemes. 
\end{itemize}

\noindent A summary of our main results is provided in Table \ref{sum_contribution} for convenience.

\begin{table}[tbh]
\centering
\caption{Variants of KYBER1024 for encrypting $\tau$ AES Keys by packing $\ell$ ciphertexts. \\{ $N$: total plaintext size (in bytes), $M$: total ciphertext size (in bytes), $\delta$: DFR, $\rho$: CER} }
\label{sum_contribution}
\centering
\vspace{-5mm}
\begin{threeparttable}
\smallskip\noindent
\resizebox{\linewidth}{!}{
\begin{tabular}{|c||c|c|c||c|c|c|}
\hline 
DFR Analysis& \multicolumn{3}{c||}{CLT}  & \multicolumn{3}{c|}{Numerical}  \\ \hline
Scheme & \cite{liu2023lattice} & \cite{LWEchannel2022}  & \cite{liu2024SCKyber}   & \cite{NISTpqcfinal2024} & \multicolumn{2}{c|}{This work}\\ \hline
Encoder & Lattice\footnotemark[2] & Q-BCH  & Binary-BCH   & Uncoded & Uncoded & Lattice\footnotemark[2] \\ \hline
 \makecell{MMSE\footnotemark[1] \\ Quantizer} & No & No   &  Yes   & No  & Yes & Yes \\ \hline
\makecell{$\tau=1$ \\($1$ AES key)} & \makecell{ $N=32$\\$ M=1184$\\ $\ell=1$  \\$\delta =2^{-213}$  \\  $\rho =37$} & \makecell{ $N=58$ \\ $ M=1568$ \\ $\ell=1$ 
 \\$\delta <2^{-174}$  \\  $\rho =26.6$} & -   & \makecell{ $N=32$\\ $ M=1568$\\ $\ell=1$ \\$\delta =2^{-174}$  \\  $\rho =49$} & \makecell{ $N=32$\\ $ M=1568$ \\ $\ell=1$  \\ $\delta =2^{-190}$  \\  $\rho =49$}  & \makecell{ $N=32$\\ $ M=1408$ \\ $\ell=8$ \\ $\delta =2^{-204}$  \\  $\rho =44$} \\ \hline
\makecell{$\tau=2$ \\ ($2$ AES keys)} & - & - & \makecell{ $N=79$\\ $ M=1792$ \\ $\ell=1$ \\$\delta =2^{-174}$  \\  $\rho =22.5$}   & - & \makecell{ $N=64$\\ $ M=1728$ \\ $\ell=2$ \\ $\delta =2^{-189}$  \\  $\rho =27$}  & \makecell{ $N=64$\\ $ M=1536$ \\ $\ell=8$ \\ $\delta =2^{-203}$  \\  $\rho =24$} \\ \hline
\makecell{$\tau=8$  \\ ($8$ AES keys)}  & - & \makecell{ $N=276$\\ $M=6272$ \\ $\ell=4$ \\$\delta <2^{-174}$  \\  $\rho =22.7$} & - & - & \makecell{ $N=256$\\ $ M=2688$ \\ $\ell=8$ \\$\delta =2^{-187}$ \\  $\rho =10.5$} & \makecell{ $N=256$\\$ M=2688$ \\ $\ell=8$ \\ $\delta =2^{-336}$ \\  $\rho =10.5$}\\ \hline
\end{tabular}}
 \begin{tablenotes}
       \item [1] MMSE quantization is defined in Definition \ref{def:MMSE}.
       \item [2] Lattice coding principles are detailed in Section~\ref{Sec:Lattice_def}. The encoding scheme for $1$–$2$ AES keys is described in Section~\ref{sec: TLC}, while those for $8$–$36$ AES keys are presented in Section~\ref{sec: LC}.
     \end{tablenotes}
\end{threeparttable}
\vspace{-6mm}
\end{table}

\section{Preliminaries}
In this section, we set the notations, provide the definitions and background on coding techniques. We further provide Kyber algorithms and identify the gaps in analysis regarding the independence assumptions used in central limit theorem (CLT) in various prior works.
\subsection{Notation and Definitions} \label{sec:nd}
\vspace{-1mm}
\emph{Rings:} Let $R_q$ denote the polynomial ring  $\mathbb{Z}_q[X]/(X^{n}+1)$, where  $n = 256$  and  $q = 3329$  in this setting. Elements of $R_q$ are represented by regular font letters, while vectors of coefficients in  $R_q$  are denoted by bold lowercase letters. Matrices and vectors are indicated by bold uppercase and lowercase letters, respectively. The transpose of a vector $\mathbf{v}$ or a matrix  $\mathbf{A}$  is represented as $\mathbf{v}^T$  or  $\mathbf{A}^T$ , respectively. By default, vectors are treated as column vectors.

\emph{Sampling and Distribution:} For a set \( \mathcal{S} \), we use the notation \( s \leftarrow \mathcal{S} \) to indicate that \( s \) is chosen uniformly at random from \( \mathcal{S} \). If \( \mathcal{S} \) represents a probability distribution, this means \( s \) is chosen according to that distribution. This notation is extended coefficient-wise to a polynomial \( f(x) \in R_q \) or a vector of such polynomials. Let \( x \) be a bit string and \( S \) be a distribution that takes \( x \) as input. We express \( y \sim S \coloneqq \mathsf{Sam}(x) \) to mean that the output \( y \) generated by the distribution \( S \) using input \( x \) can be extended to any desired length. We define \( \beta_{\eta} = B(2\eta, 0.5) - \eta \) as the central binomial distribution over \( \mathbb{Z} \). The Cartesian product of two sets \( A \) and \( B \) is represented as \( A \times B \). We denote \( A \times A \) as \( A^2 \).

\emph{Compression and Quantization:} Given \( x \in \mathbb{R} \), the notation \( \left\lceil x \right\rfloor \) refers to rounding \( x \) to the nearest integer, with ties rounded up. The operations \( \left\lfloor x \right\rfloor \) and \( \left\lceil x \right\rceil \) denote rounding \( x \) down and up, respectively. Now, considering \( x \in \mathbb{Z}_q \) and \( d \in \mathbb{Z} \) such that \( 2^d < q \),  Kyber compression and decompression functions are \cite{NISTpqcfinal2024}:
\begin{align}
x'&=\mathsf{Compress}_{q}(x,d)=\lceil (2^{d}/q)\cdot x\rfloor \mod 2^{d},\nonumber\\
\hat{x}&=\mathsf{Decompress}_{q}(x',d)=\lceil (q/2^{d})\cdot x'\rfloor \in \mathcal{C}.\label{ComDecom}
\end{align}
Kyber compression and decompression operations can be interpreted as a mapping from a large set $\mathbb{Z}_q$ to a smaller set $\mathcal{C}$ with $|\mathcal{C}|=2^d < q$. In the literature of signal processing, this mapping is generally known as \emph{quantization}.

\begin{definition}[Scalar Quantization] \label{def:quan} Given a random variable \( x \in \mathbb{Z}_q \) and an integer \( L > 0 \), a scalar quantization \( Q_L \) divides the support of \( x \) into \( L \) subsets \( R_1, \ldots, R_L \), referred to as quantization regions \( T_L = \bigcup_{i=1}^L R_i \). Each region \( R_j \) is associated with a quantizer \( \alpha_j \in \mathcal{C}_L \). When \( x \) lies within the region \( R_j \), the quantization \( Q_L \) maps \( x \) to the point \( \hat{x} = \alpha_j \). $Q_L$ can be viewed as a function:
\begin{equation}
Q_L: \mathbb{Z}_{q} \rightarrow \mathcal{C}_L,
\end{equation}
where $Q_{L}({x}, \mathcal{C}_L, T_L):=\hat{x}$ can be uniquely represented by its index in $\mathcal{C}_L$, denoted as $\mathsf{Index}_L(\hat{{x}})$, i.e., $\mathcal{C}_L(\mathsf{Index}_L(\hat{\mathbf{x}})) = \hat{\mathbf{x}}$.
The communication cost of transmitting $\hat{x}$ reduces to $\log_2(L)$ bits.
\end{definition}

For consistency, with $L=2^d$, the Kyber quantization in \eqref{ComDecom} is redefined as
\begin{align}
\hat{x}&=Q_{\mathsf{Kyber}, 2^d}({x})=\mathsf{Decompress}_{q}(\mathsf{Compress}_{q}(x,d),d) \nonumber \\
x'&=\mathsf{Index}_{2^d}(\hat{{x}})=\mathsf{Compress}_{q}(x,d). \label{k_quan}
\end{align}

\begin{definition}[MMSE Quantization] \label{def:MMSE} 
The optimal quantization should minimize the mean squared quantization error (MMSE):
\begin{equation}
(\mathcal{C}_L, T_L) = \arg \min_{\mathcal{C}'_L \in \mathbb{R}^n, T_L'\subset \mathbb{R}^n }\mathsf{E}(\|\mathbf{x}-Q_{L}(\mathbf{x}, \mathcal{C}'_L, T_L')\|^2). \label{MMSE}
\end{equation}
For simplicity of notation, we define the MMSE quantization as
\begin{equation}
\hat{\mathbf{x}}=Q_{\mathsf{MMSE},L}(\mathbf{x}).
\end{equation}
\end{definition}
The MMSE scalar quantization is the Lloyd-Max quantization \cite{Lloy1982}.

\begin{definition}[Moment Generating Function]\label{def:MGF}
Let $X \leftarrow D$ be a random variable. For $\theta \in \mathbb{R}$, the moment generating function (MGF) of $X$ is denoted by
\begin{equation}
    M_{X}(\theta)=\E(\exp(\theta X)). \label{MGF}
\end{equation}
\end{definition}

\begin{definition}[Algebraic Expression of a Column Ring Vector]\label{def:ME}
A column ring vector \( \mathbf{v} \in R_q^\ell \) is defined as:
\[
\mathbf{v} =[v_0(x),v_1(x),\ldots,v_{\ell-1}(x)]^T,
\; \text{where } v_i(x) = \sum_{j=0}^{n-1} v_{i,j} x^j \in R_q, \quad v_{i,j} \in \mathbb{Z}_q.
\]
Define the mapping function $\phi: R_q^\ell \to \mathbb{Z}_q^{\ell \times n}$:
\[
\phi(\mathbf{v}) =
\begin{bmatrix}
v_{0,0} & v_{0,1} & \cdots & v_{0,n-1} \\
v_{1,0} & v_{1,1} & \cdots & v_{1,n-1} \\
\vdots & \vdots & \ddots & \vdots \\
v_{\ell-1,0} & v_{\ell-1,1} & \cdots & v_{\ell-1,n-1}
\end{bmatrix}.
\]
The function \( \phi \) extracts the coefficients of each ring element \( v_i(x) \) and arranges them as the \( i \)-th row of \( \phi(\mathbf{v}) \in \mathbb{Z}_q^{k \times n} \). The inverse mapping \( \phi^{-1} : \mathbb{Z}_q^{\ell \times n} \rightarrow R_q^\ell \) reconstructs the ring vector from its coefficient matrix: $\mathbf{v}=\phi^{-1}(\phi(\mathbf{v}))$.
\end{definition}

\subsection{Lattice Code, Encoder, and Decoder} \label{Sec:Lattice_def}
\begin{definition}[Lattice]
An \(\ell\)-dimensional lattice \(\mathsf{\Lambda}\) is a discrete additive subgroup of \(\mathbb{R}^M\) with \(M \geq \ell\). Given \(\ell\) linearly independent column vectors \(\mathbf{b}_1, \ldots, \mathbf{b}_\ell \in \mathbb{R}^M\), the lattice generated by these vectors is defined as:
\begin{equation*}
\mathsf{\Lambda} =\mathcal{L}(\mathbf{B})=  \left \{\sum_{i=1}^{\ell} z_i\mathbf{b}_i|z_i\in\mathbb{Z} \right\},
\end{equation*}
where $\mathbf{B}=[\mathbf{b}_{1},\ldots,%
\mathbf{b}_{\ell}]$ is referred to as a generator matrix of $\mathsf{\Lambda}$.
\end{definition}

\begin{definition}[Lattice Code]
A lattice code \(\mathcal{C}(\mathsf{\Lambda}, \mathcal{P})\) is the finite collection of points in \(\mathsf{\Lambda}\) that fall within the bounded set \(\mathcal{P}\):
\begin{equation*}
\mathcal{C}(\mathsf{\Lambda}, \mathcal{P}) = \mathsf{\Lambda} \cap \mathcal{P}.
\end{equation*}
If $\mathcal{P}=\mathbb{Z}_p^\ell$, the code $\mathcal{C}(\mathsf{\Lambda}, \mathbb{Z}_p^\ell)$ is said to be generated from hypercube shaping (HS).
\end{definition}

\begin{definition}[CVP Decoder]
Given \(\mathbf{y} \in \mathbb{R}^\ell\), the Closest Vector Problem (CVP) decoder returns the nearest lattice vector to \(\mathbf{y}\) within the lattice \(\mathcal{L}(\mathbf{B})\):
\begin{equation*}
\mathbf{x}=\mathsf{CVP}(\mathbf{y}, \mathcal{L}(\mathbf{B}))=\arg\min_{\mathbf{x}' \in \mathcal{L}(\mathbf{B})} \| \mathbf{x}'-\mathbf{y} \|.
\end{equation*}
\end{definition}

\begin{definition}[HS Encoder \cite{FrodoCong2022}]\label{def:HS}
Let \(\mathbf{B} = \mathbf{U} \cdot \diag(\pi_1, \ldots, \pi_\ell) \cdot \mathbf{U}'\) be the Smith Normal Form (SNF) factorization of a lattice basis \(\mathbf{B}\), where \(\mathbf{U}\) and \(\mathbf{U}'\) are unimodular matrices in \(\mathbb{Z}^{\ell \times \ell}\). Let the message space be
\begin{equation}
\mathcal{M}_{p,\ell}=\left\{ 0,1,\ldots,p/\pi_1-1\right\} \times \cdots \times
\left\{ 0,1,\ldots,p/\pi_\ell-1\right\} , \label{m_p}
\end{equation}%
where $p > 0$ is a common multiple of $\pi_1, \ldots, \pi_\ell$. Given an input \(\mathbf{m} \in \mathcal{M}_{p, \ell}\) and $\hat{\mathbf{B}}= \mathbf{U}\cdot \diag (\pi_1, \ldots, \pi_\ell)$, a HS encoder produces a codeword \(\mathbf{x} \in \mathcal{C}(\mathcal{L}(\mathbf{B}), \mathbb{Z}_p^\ell)\):
\begin{equation}
\mathbf{x} = \hat{\mathbf{B}}\mathbf{m} \bmod p , \label{HC_shaping}
\end{equation}
\end{definition}

\begin{definition}[HS CVP Decoder \cite{FrodoCong2022}]\label{def:HS-CVP}
Given a lattice \(\mathcal{L}(\mathbf{B})\) in \(\mathbb{R}^M\) and an input vector \(\mathbf{y} \in \mathbb{R}^M\), the HS CVP decoder outputs an estimated message \(\hat{\mathbf{m}} = [\hat{m}_1, \ldots, \hat{m}_\ell]^T \in \mathcal{M}_{p, \ell}\):
\begin{align}
\hat{\mathbf{m}} &=\mathsf{{CVP}_{HS}}(\mathbf{y}, \mathcal{L}(\mathbf{B}))  \notag \\ &=\hat{\mathbf{B}}^{-1}\cdot\mathsf{CVP}(\mathbf{y}, \mathcal{L}(\mathbf{B})) \bmod (p/\pi_1, \ldots, p/\pi_\ell), \label{HC_decoding}
\end{align}%
where $\hat{m}_i= (\hat{\mathbf{B}}^{-1}\mathsf{CVP}(\mathbf{y}, \mathcal{L}(\mathbf{B}))_i \bmod p/\pi_i$, for $i= 1,\ldots, \ell$.
\end{definition}

\subsection{Cryptographic Definitions}

\begin{definition}[M-LWE Problem \cite{Kyber2018}]\label{def:MLWE_P}
The M-LWE samples \( (\mathbf{a}_i, b_i = \mathbf{a}_i^T \mathbf{s} + e_i) \) are drawn from the M-LWE distribution \( A_{\mathbf{s}, \beta} \) over \( R_q^k \times R_q \). Here, \(\mathbf{a}_i \leftarrow R_q^k\) is chosen uniformly, \(\mathbf{s} \leftarrow \beta_\eta^k\) is common to all samples, and \(e_i \leftarrow \beta_\eta\) is independent for each sample. Given \(m\) M-LWE samples, the decision-M-LWE problem involves distinguishing \( A_{\mathbf{s}, \beta} \) from the uniform distribution on \( R_q^k \times R_q\), while the search-M-LWE problem seeks to recover the secret \(\mathbf{s}\). For an algorithm $\mathsf{A}$, we define the advantage of an adversary as $\mathsf{Adv}^{\mathsf{M-LWE}}_{m, k, \eta}(\mathsf{A}) =$
\begin{equation*}
\left |\Pr \left(b'=1 : \begin{array}{c}
\mathbf{A} \leftarrow R_q^{m \times k} ; (\mathbf{s}, \mathbf{e})\leftarrow \beta^{k}_\eta \times \beta^{m}_\eta\\
\mathbf{b}=\mathbf{A}\mathbf{s} +\mathbf{e}; b' \leftarrow \mathsf{A}(\mathbf{A}, \mathbf{b}) \end{array}\right) - \Pr \left( b'=1 : \begin{array}{c} 
\mathbf{A} \leftarrow R_q^{m \times k} ; \\
\mathbf{b} \leftarrow R_q^{m}; b' \leftarrow \mathsf{A}(\mathbf{A}, \mathbf{b}) \end{array} \right) \right |
\end{equation*}
\end{definition}

\begin{definition}[Public-Key Encryption (PKE) \cite{Kyber2018}]
A public-key encryption scheme \( \mathsf{PKE} = (\mathsf{KeyGen}, \mathsf{Enc}, \mathsf{Dec}) \) consists of a triple of probabilistic algorithms along with a message space \(\mathcal{M}\). The key-generation algorithm \(\mathsf{KeyGen}\) produces a pair \((pk, sk)\), which includes a public key and a secret key. The encryption algorithm \(\mathsf{Enc}\) takes the public key \(pk\) and a message \(m \in \mathcal{M}\) to generate a ciphertext \(c\). Finally, the deterministic decryption algorithm \(\mathsf{Dec}\) uses the secret key \(sk\) and the ciphertext \(c\) to output either a message \(m \in \mathcal{M}\) or a special symbol \(\bot\) to indicate rejection. We say that the scheme is \((1-\delta)\)-correct if $\E[\max_{m \in \mathcal{M}} \Pr[\mathsf{Dec}(sk, \mathsf{Enc}(pk,m))=m]] \geq 1 -\delta$, where the expectation is taken over $(pk, sk)$ and the probability is taken over the random coins of $\mathsf{Enc}$.
\end{definition}

\begin{definition}[IND-CPA and IND-CCA \cite{Kyber2018}]\label{def:cpacca}
We revisit the standard security notions for public-key encryption, specifically indistinguishability under chosen-ciphertext attacks (IND-CCA) and chosen-plaintext attacks (IND-CPA). The advantage of
an adversary $\mathsf{A}$ is defined as 
\begin{equation}
\mathsf{Adv}^{\mathsf{CCA}}_{\mathsf{PKE}}(\mathsf{A}) = \left |\Pr \left(b=b' : \begin{array}{c}
(\mathsf{pk}, \mathsf{sk}) \leftarrow \mathsf{KeyGen}();\\ (m_0,m_1, s) \leftarrow \mathsf{A}^{\mathsf{DEC}(\cdot)}(\mathsf{pk});\\
b \leftarrow \{0,1\}; c^{*} \leftarrow \mathsf{Enc}(\mathsf{pk}, m_b); \\
b' \leftarrow \mathsf{A}^{\mathsf{DEC}(\cdot)}(s, c^{*});
\end{array}\right) - 1/2 \right |
\end{equation}
where the decryption oracle is defined as $\mathsf{DEC}(\cdot) = \mathsf{Dec}(\mathsf{sk},\cdot)$. We also require that \( |m_0| = |m_1| \) and that in the second phase, the adversary \(\mathsf{A}\) is not permitted to query \(\mathsf{DEC}(\cdot)\) with the challenge ciphertext \(c^{*}\). The advantage \(\mathsf{Adv}^{\mathsf{CPA}}_{\mathsf{PKE}}(\mathsf{A})\) of an adversary \(\mathsf{A}\) is defined as \(\mathsf{Adv}^{\mathsf{CCA}}_{\mathsf{PKE}}(\mathsf{A})\), provided that \(\mathsf{A}\) cannot query $\mathsf{DEC}(\cdot)$.
\end{definition}

\begin{definition}[DFR and CER] Given a message $ m \in \mathcal{M}$, the Decryption Failure Rate (DFR) is denoted as $\delta \triangleq \Pr (\hat{m} \neq m)$, where $\hat{m}$ is the decryption of $c$ where $c= \mathsf{Enc}(pk,m)$. The communication cost refers to the ciphertext expansion rate (CER):
\begin{equation}
\rho =  \frac{\# \text{ of bits in } c}{\# \text{ of bits in } m}, \label{CER}
\end{equation}
i.e., the ratio of the ciphertext size to the plaintext size. 
\end{definition}

\vspace{-3mm}
\subsection{Kyber's IND-CPA-Secure Encryption}
Each message $m \in \{0,1\}^{n}$ can be
viewed as a polynomial in $R$ with coefficients in $\{0,1\}$. We recall Kyber.CPA = (KeyGen; Enc; Dec) \cite{Kyber2021} as described in Algorithms \ref{alg:kyber_keygen} to \ref{alg:kyber_dec}. The values of $\delta$, CER, and $(q, k, \eta_1,\eta_2, d_u, d_v)$ are given in Table \ref{Kyber_Par}. Note that the parameters $(q, k, \eta_1,\eta_2)$ determine the security level of Kyber, while the parameters $(d_u, d_v)$ describe the ciphertext compression rate.

\vspace{-5mm}
\setcounter{algorithm}{0}
\begin{algorithm}[H]
\caption{$\mathsf{Kyber.CPA.KeyGen()}$: key generation}
\label{alg:kyber_keygen}
\begin{algorithmic}[1]

    \State
    $\psi,\sigma\leftarrow\left\{ 0,1\right\} ^{256}$

    \State
    $\bA\sim R_{q}^{k\times k}\coloneqq\mathsf{Sam}(\psi)$

    \State
    $(\bs,\be)\sim\beta_{\eta_1}^{k}\times\beta_{\eta_1}^{k}\coloneqq\mathsf{Sam}(\sigma)$

    \State
    $\bt\coloneqq\boldsymbol{\mathrm{As+e}}$\label{line:t}

    \State \Return $\left(pk\coloneqq(\boldsymbol{\mathrm{t}},\psi),sk\coloneqq\bs\right)$  

\end{algorithmic}
\end{algorithm}

\vspace{-12mm}

\begin{algorithm}[H]
\caption{$\mathsf{Kyber.CPA.Enc}$ $(pk=(\boldsymbol{\mathrm{t}},\psi),m\in\{0,1\}^{n})$}
\label{alg:kyber_enc}
\begin{algorithmic}[1]

	\State
	$r \leftarrow \{0,1\}^{256}$

	\State
	$\boldsymbol{\mathrm{A}}\sim R_{q}^{k\times k}\coloneqq\mathsf{Sam}(\psi)$
	
	\State  $(\boldsymbol{\mathrm{r}},\boldsymbol{\mathrm{e}_{1}},e_{2})\sim\beta_{\eta_1}^{k}\times\beta_{\eta_2}^{k}\times\beta_{\eta_2}\coloneqq\mathsf{Sam}(r)$

    \State  $\boldsymbol{\mathrm{u}}\coloneqq Q_{\mathsf{Kyber}, 2^{d_{u}}}(\buone)$\label{line:u}

 \State  $v\coloneqq Q_{\mathsf{Kyber}, 2^{d_{v}}}(\boldsymbol{\mathrm{t}}^{T}\boldsymbol{\mathrm{r}}+e_2+\left\lceil {q}/{2}\right\rfloor \cdot m)$\label{line:v}

\State \Return $c\coloneqq(\mathsf{Index}_{2^{d_{u}}}(\boldsymbol{\mathrm{u}}),\mathsf{Index}_{2^{d_{v}}}(v))$

\end{algorithmic}
\end{algorithm}

\vspace{-12mm}

\begin{algorithm}[H]
\caption{${\mathsf{Kyber.CPA.Dec}}\ensuremath{(sk=\bs,c=(\bu,v))}$}
\label{alg:kyber_dec}
\begin{algorithmic}[1]

    \State
    $\bu\coloneqq \mathcal{C}_{2^{d_{u}}}(\mathsf{Index}_{2^{d_{u}}}(\boldsymbol{\mathrm{u}}))$

    \State
    $v\coloneqq \mathcal{C}_{2^{d_{v}}}(\mathsf{Index}_{2^{d_{v}}}(v))$

    \State \Return $\mathsf{Compress}_{q}(v-\bs^{T}\bu,1)$

\end{algorithmic}
\end{algorithm}

\vspace{-10mm}

\begin{table}[ht]
\caption{Parameters of Kyber \cite{NISTpqcfinal2024}}
\label{Kyber_Par}\centering
\vspace{-3mm}
\begin{tabular}{|c||c|c|c|c|c|c|c|c|c|}
\hline
& $k$ & $q$ & $\eta_{1}$ & $\eta_{2}$ & $d_{u}$ & $d_{v}$ & DFR & CER & Plaintext Size \\ \hline
KYBER512 &  $2$ & $3329$ & $3$ & $2$ & $10$ & $4$ & $2^{-138}$ & $24$ & $256$ bits\\ \hline
KYBER768 & $3$ & $3329$ & $2$ & $2$ & $10$ & $4$ & $2^{-164}$ &$34$ & $256$ bits\\ \hline
KYBER1024 & $4$ & $3329$ & $2$ & $2$ & $11$ & $5$ & $2^{-174}$ &$49$ & $256$ bits \\ \hline
\end{tabular}
\vspace{-3mm}
\end{table}

\subsection{Kyber with Optimal Quantization}
\vspace{-1mm}

The choice of quantization affects the distribution of \((c_v, \mathbf{c}_u)\) and thus the decoding noise. \cite{liu2024SCKyber} shows that Kyber’s quantizer \(Q_{\mathsf{Kyber}}\) is suboptimal, as it does not minimize the mean squared values of \((c_v, \mathbf{c}_u)\). Replacing it with the MMSE-optimal Lloyd-Max quantizer \(Q_{\mathsf{MMSE}}\) improves the DFR without affecting Kyber’s security, which is independent of the quantization method and noise level.

Table~\ref{Com_KC_MMSE} presents DFR bounds under different quantizers. Using code from \cite{KyberCode} and the Lloyd-Max noise distribution, we compute the DFR numerically. A noticeable gap emerges between our results and the CLT-based asymptotic bound in \cite{liu2024SCKyber}, indicating that CLT-based analyses may underestimate DFR. Still, both bounds confirm that improved quantization reduces DFR.

\begin{table}[ht]
\caption{DFR bounds: Kyber Compression $Q_{\mathsf{Kyber}}$ vs. Lloyd-Max $Q_{\mathsf{MMSE}}$}
\label{Com_KC_MMSE}\centering
\vspace{-3mm}
\small
\begin{tabular}{|c||c|c||c|}
\hline
Source & \cite{NISTpqcfinal2024} & \cite{liu2024SCKyber}  & This work \\ \hline
 Quantization & $Q_{\mathsf{Kyber}}$ &  $Q_{\mathsf{MMSE}}$  &  $Q_{\mathsf{MMSE}}$  \\ \hline
 Bound type & Numerical & CLT  &  Numerical \\ \hline
KYBER512    & $2^{-138}$ & $2^{-150}$  & $2^{-142}$ \\ \hline
KYBER768    & $2^{-164}$ & $2^{-177}$  & $2^{-169}$  \\ \hline
KYBER1024    & $2^{-174}$ & $2^{-196}$  & $2^{-190}$  \\ \hline
\end{tabular}
\vspace{-3mm}
\end{table}

\subsection{Coded Kyber and CLT Assumption}
Kyber decryption decoding problem can be expressed as \cite{Kyber2021}%
\begin{equation}
y=v-\mathbf{s}^{T}\mathbf{u}=\left\lceil q/2\right\rfloor \cdot m+n_{e} \text{,}  \label{decoding_mode}
\end{equation}
where $n_{e}$ is the decryption decoding noise
\begin{align}
n_{e}&=v-\mathbf{s}^{T}\mathbf{u} -\left\lceil q/2\right\rfloor
\cdot m  = \mathbf{e}^{T}\mathbf{r}+e_{2}+c_{v}-\mathbf{s}^{T}\left( \mathbf{e}%
_{1}+\mathbf{c}_{u}\right),  \label{Ne}
\end{align}
where $(c_{v},\mathbf{c}_{u})$ refers to the quantization noises produced by the quantization $Q_{\mathsf{Kyber}}$ in Algorithm \ref{alg:kyber_enc}. Due to $n_e$, Kyber decryption has a failure rate.

From the information theory perspective, (\ref{decoding_mode}) can be viewed as an \emph{uncoded} $2$-PAM \cite{Proakis}, which has been generalized to the coded cases  \cite{liu2023lattice}\cite{liu2024SCKyber}\cite{LWEchannel2022}:
\begin{equation}
y=\left\lceil q/p\right\rfloor \cdot \mathsf{ENC}(m)+n_{e} \text{,}  \label{decoding_mode_coded}
\end{equation}
where $p \in \mathbb{Z}$, $m \leftarrow \{0,1\}^K$, and $\mathsf{ENC}(m) :  \{0,1\}^K \rightarrow \mathbb{Z}_p^{n}$ represents an encoder. For example, \cite{liu2023lattice} uses a lattice encoder, \cite{liu2024SCKyber} uses a binary BCH encoder, and \cite{LWEchannel2022} uses a Q-ary BCH encoder. The advantage of coded Kyber is the reduced CER, since more information bits can be encrypted to a single ciphertext.

\vspace{3mm}

\noindent {\bf Independence/CLT Assumption in \cite{liu2023lattice}\cite{liu2024SCKyber}\cite{LWEchannel2022}:} To estimate the DFR of coded Kyber, existing schemes assume that the entries in $n_e$ are mutually independent. This assumption relies on the CLT, i.e., for a certain variance $\sigma_e^{2}$,
\begin{equation}
\mathbf{e}^{T}\mathbf{r}-\mathbf{s}^{T}\left( \mathbf{e}%
_{1}+\mathbf{c}_{u}\right) \rightarrow \mathcal{N}(0, \sigma_e^{2}\mathbf{I}_n), \text{ as } k \cdot n \rightarrow \infty \label{CLT_A}
\end{equation}
An open question is whether we can develop a coded Kyber scheme without relying on the CLT assumption on $n_e$. We will address this question in the remainder of the paper.


\vspace{-0mm}

\section{Uncoded P$_\ell$-Kyber: Kyber with Packed Ciphertexts} \label{Sec: uncoded}
In this section, we first present an $\ell$-layer Kyber following PVW approach~\cite{PVW2008}. We then turn this into an IND-CCA KEM and analyze its key and ciphertext sizes as well as its DFR and CER compared to original Kyber.
\subsection{IND-CPA-Secure Encryption}
We consider the idea of \emph{packed ciphertexts} in \cite{PVW2008}, where  a ciphertext $c$ encrypts a vector of $\ell$ plaintext ring elements $\mathbf{m}=[m_1,\ldots, m_\ell]^T \in R_2^{\ell}$, not just a single ring element $m \in R_2$. In details, the same matrix $\mathbf{A}$ and encryption randomness $\mathbf{r}$ in Algorithm \ref{alg:kyber_enc} can be securely reused to encrypt $\mathbf{m}$, by having $\ell$ secret-key vectors $\mathbf{S}=[\mathbf{s}_1,\ldots,\mathbf{s}_\ell]$. The key generation, encryption and decryption functions of Kyber with $\ell$-packed ciphertexts (P$_\ell$-Kyber PKE) is given below.

\setcounter{algorithm}{0}
\vspace{-3mm}
\begin{algorithm}[H]
\caption{$\mathsf{Kyber.Packed.CPA.KeyGen()}$: key generation}
\label{alg:kyberP_keygen}
\begin{algorithmic}[1]

    \State
    $\psi,\sigma\leftarrow\left\{ 0,1\right\} ^{256}$

    \State
    $\bA\sim R_{q}^{k\times k}\coloneqq\mathsf{Sam}(\psi)$

    \State
    $(\mathbf{S},\mathbf{E})\sim\beta_{\eta_1}^{k \times \ell}\times\beta_{\eta_1}^{k \times \ell}\coloneqq\mathsf{Sam}(\sigma)$

    \State
    $\mathbf{T} \coloneqq\boldsymbol{\mathrm{AS+E}}$

    \State \Return $\left(pk\coloneqq(\boldsymbol{\mathrm{T}},\psi),sk\coloneqq \mathbf{S} \right)$  

\end{algorithmic}
\end{algorithm}

\vspace{-12mm}

\begin{algorithm}[H]
\caption{$\mathsf{Kyber.Packed.CPA.Enc}$ $(pk=(\boldsymbol{\mathrm{T}},\psi),\mathbf{m} \in \{0,1\} ^{\ell \times n}$)}
\label{alg:kyberP_enc}
\begin{algorithmic}[1]

	\State
	$r \leftarrow \{0,1\}^{256}$

	\State
	$\boldsymbol{\mathrm{A}}\sim R_{q}^{k\times k}\coloneqq\mathsf{Sam}(\psi)$
	
	\State  $(\boldsymbol{\mathrm{r}},\boldsymbol{\mathrm{e}_{1}},\mathbf{e}_{2})\sim\beta_{\eta_1}^{k}\times\beta_{\eta_2}^{k}\times\beta_{\eta_2}^{\ell}\coloneqq\mathsf{Sam}(r)$
	
	\State  $\boldsymbol{\mathrm{u}}\coloneqq Q_{\mathsf{MMSE}, 2^{d_{u}}}(\buone)$
	
	\State  $\mathbf{v}\coloneqq Q_{\mathsf{MMSE}, 2^{d_{v}}}(\boldsymbol{\mathrm{T}}^{T}\boldsymbol{\mathrm{r}}+\mathbf{e}_2+\left\lceil {q}/{2}\right\rfloor \cdot \mathbf{m})$
	
	\State \Return $c\coloneqq(\mathsf{Index}_{2^{d_{u}}}(\boldsymbol{\mathrm{u}}),\mathsf{Index}_{2^{d_{v}}}(\mathbf{v}))$

\end{algorithmic}
\end{algorithm}

\vspace{-12mm}

\begin{algorithm}[H]
\caption{${\mathsf{Kyber.Packed.CPA.Dec}}\ensuremath{(sk=\mathbf{S},c=(\bu,\mathbf{v}))}$}
\label{alg:kyberP_dec}
\begin{algorithmic}[1]

    \State
    $\bu\coloneqq \mathcal{C}_{2^{d_{u}}}(\mathsf{Index}_{2^{d_{u}}}(\boldsymbol{\mathrm{u}}))$

    \State
    $\mathbf{v}\coloneqq \mathcal{C}_{2^{d_{v}}}(\mathsf{Index}_{2^{d_{u}}}(\boldsymbol{\mathrm{v}}))$

    \State \Return $\mathsf{Compress}_{q}(\mathbf{v}-\mathbf{S}^{T}\bu,1)$

\end{algorithmic}
\end{algorithm}
\vspace{-3mm}

\noindent {\bf Correctness.} Let $\delta_\ell$ be the DFR of P$_\ell$-Kyber PKE. We show below the correctness of the encryption scheme described in Algorithms \ref{alg:kyberP_keygen} to \ref{alg:kyberP_dec}. 

\begin{lemma}[Correctness of P$_\ell$-Kyber PKE]
The DFR is bounded by
\begin{equation}
\delta_{\ell}\leq \ell \cdot \delta,
\end{equation}
where $\delta$ is the DFR of the unpacked Kyber in Table \ref{Com_KC_MMSE}.
\end{lemma}
\begin{proof}
Let $\mathbf{n}_{e} =[n_{0},\ldots,n_{\ell-1}]^T \in R_q^{\ell}$ be the decoding noise of P$_\ell$-Kyber. Similar to \eqref{Ne}, we can write $\mathbf{n}_{e}$ as%
\begin{align}
\mathbf{n}_{e}&=\mathbf{v}-\mathbf{S}^{T}\mathbf{u} -\left\lceil q/2\right\rfloor
\cdot \mathbf{m} = \mathbf{E}^{T}\mathbf{r}+\mathbf{e}_{2}+\mathbf{c}_{v}-\mathbf{S}^{T}\left( \mathbf{e}%
_{1}+\mathbf{c}_{u}\right),
\end{align}
where $(\mathbf{c}_{v},\mathbf{c}_{u})$ are the quantization noises. Using the union bound, we obtain
\begin{equation}
\delta_\ell = \Pr (\|\mathbf{n}_e\|_{\infty} \geq \lceil q/4 \rfloor) \leq \sum_{i=0}^{\ell-1}\Pr(\|n_{i} \|_{\infty}\geq \lceil q/4 \rfloor) = \ell \cdot \delta.
\end{equation}
\end{proof}

\begin{remark}
The DFR of P$_\ell$-Kyber increases with $\ell$ but can be reduced using \(Q_{\mathsf{MMSE}}\). Its key benefit is the low CER, denoted as $\rho_\ell$: 
\begin{equation}
    \rho_\ell = \dfrac{knd_u+\ell n d_v}{N} = \dfrac{kd_u}{\ell}+ d_v. \label{CER_P}
\end{equation}
where $N=n \cdot \ell$ is the plaintext size (in bits). Table \ref{sum_contribution} shows \((\delta=\delta_\ell, \rho=\rho_\ell)\) as a function of \(\ell\) (refer to the column labeled ``This Work – Uncoded''). P$_{8}$-KYBER1024 reduces CER by $79\%$ and DFR by \(2^{13}\), relative to KYBER1024.
\end{remark}

\noindent {\bf Security.} We will prove that the
encryption scheme defined above is IND-CPA secure under the M-LWE hardness assumption.

\begin{lemma}[IND-CPA Security of P$_\ell$-Kyber PKE]
For any adversary $\mathsf{A}$, there exists an adversary $\mathsf{B}$
such that $\mathsf{Adv}^{\mathsf{CPA}}_{\mathsf{P}_\ell-\mathsf{Kyber}}( \mathsf{A}) \leq (\ell+1) \cdot \mathsf{Adv}^{\mathsf{M-LWE}}_{k+\ell, k, \eta}(\mathsf{B})$.
\end{lemma}
\begin{proof}
Let $\mathsf{A}$ be an adversary that is executed in the IND-CPA
security experiment which we call game $G_0$, i.e., $\mathsf{Adv}^{\mathsf{CPA}}_{\mathsf{P}_\ell-\mathsf{Kyber}} =|\Pr(b = b' \text{ in game } 
 G_0)- 1/2|$. 

In game $G_1$, the $\ell$ column vectors in the public key $\mathbf{T}$ are simultaneously substituted with $\ell$ uniform random vectors. It is possible to verify that there exists an adversary $\mathsf{B}$ with the same running time as that of $\mathsf{A}$ such that 
\begin{equation}~\label{eq:G0G1l}
|\Pr(b = b' \text{ in game } G_0)- \Pr(b = b' \text{ in game } G_{1})| \leq \ell \mathsf{Adv}^{\mathsf{M-LWE}}_{k, k, \eta}(\mathsf{B}) \leq \ell \mathsf{Adv}^{\mathsf{M-LWE}}_{k+\ell, k, \eta}(\mathsf{B}), 
\end{equation}
where the second inequality holds since the adversary $\mathsf{B}$ will have access to more samples, in particular from $k$ to $k+\ell$.

In game $G_2$, the vectors $\mathbf{u}$ and $\mathbf{v}$ used in the generation of the challenge ciphertext are simultaneously substituted with uniform random vectors. Again, there exists an adversary $\mathsf{B}$ with the same running time as that of $\mathsf{A}$ with
\begin{equation}~\label{eq:Advmlwe}
|\Pr(b = b' \text{ in game } G_{1})- \Pr(b = b' \text{ in game } G_2)| \leq \mathsf{Adv}^{\mathsf{M-LWE}}_{k+\ell, k, \eta}(\mathsf{B}). 
\end{equation}
Note that in game $G_2$, the value $\mathbf{v}$ from the challenge ciphertext is independent of bit $b$ and therefore $\Pr(b = b' \text{ in game } G_2) = 1/2$. Collecting the probabilities in \eqref{eq:G0G1l} and \eqref{eq:Advmlwe} yields the required bound.
\end{proof}

\subsection{The CCA-Secure KEM}
Let $G: \{0, 1\}^* \rightarrow \{0, 1\}^{(\ell+1) \times 256}$ and $H: \{0, 1\}^* \rightarrow \{0, 1\}^{\ell \times 256}$ be hash functions. Given $z \leftarrow \{0, 1\}^{\ell \times 256}$, along the same line as \cite{Kyber2018}, a KEM is obtained by
applying a KEM variant of the Fujisaki–Okamoto (FO) transform \cite{modFOT2017} to the P$_\ell$-Kyber encryption scheme. We make explicit the randomness ${r}$ in the Enc algorithm.

\vspace{-5mm}
\setcounter{algorithm}{0}
\begin{algorithm}[H]
\caption{ $\mathsf{Kyber.Packed.Encaps}(pk = (\boldsymbol{\mathrm{T}}, \psi)$)}
\label{alg:kyberP_encaps}
\begin{algorithmic}[1]

	\State
	$\mathbf{m} \leftarrow \{0,1\}^{256 \times \ell}$

       \State
       $(\hat{K}, {r}) \coloneqq G(H(pk), \mathbf{m})$

       \State
       $ (\mathbf{u}, \mathbf{v}) \coloneqq  \mathsf{Kyber.Packed.CPA.Enc}$ $(pk=(\boldsymbol{\mathrm{T}},\psi),\mathbf{m}; {r}$)

       \State
       $c \coloneqq (\mathbf{u}, \mathbf{v})$

       \State
       $K \coloneqq H(\hat{K}, H(c))$\footnotemark

       \State
       \Return $(c, K)$
	
\end{algorithmic}
\end{algorithm}
\footnotetext{$H(c)$ was used in \cite{Kyber2021}\cite{Kyber2018} to simplify the implementation with non-incremental hash APIs. We can use $c$ in place of $H(c)$, as referenced in \cite{modFOT2017}\cite{NISTpqcfinal2024}. Another difference between \cite{Kyber2018} and \cite{Kyber2021} \cite{NISTpqcfinal2024} is that a third hash function, Key Derivation Function (KDF), is used to compute $K$, i.e.,  $K \coloneqq \mathsf{KDF}(\hat{K}, H(c))$ in \cite{Kyber2021}. Since these small tweaks don't affect the security and DFR analysis of Kyber, we follow the original design in \cite{Kyber2018}.}

\vspace{-12mm}
\begin{algorithm}[H]
\caption{ $\mathsf{Kyber.Packed.Decaps}(sk= (\mathbf{S}, z, \mathbf{T}, \psi ),c=(\bu,\mathbf{v}))$}
\label{alg:kyberP_decaps}
\begin{algorithmic}[1]

	\State
	$\mathbf{m}' \coloneqq{\mathsf{Kyber.Packed.CPA.Dec}}{(\mathbf{S},(\bu,\mathbf{v}))}$

       \State
       $(\hat{K}',{r}') \coloneqq G(H(pk), \mathbf{m}')$

       \State
       $ (\mathbf{u}', \mathbf{v}') \coloneqq  \mathsf{Kyber.Packed.CPA.Enc}$ $(pk=(\boldsymbol{\mathrm{T}},\psi),\mathbf{m}'; {r}'$)

       \If{$(\mathbf{u}', \mathbf{v}') =(\mathbf{u}, \mathbf{v})$}

       \State 
       \Return $K \coloneqq H(\hat{K}', H(c))$
       \Else
       \State
       \Return $K \coloneqq H(z, H(c))$

       \EndIf
	
\end{algorithmic}
\end{algorithm}

\vspace{-5mm}

\noindent {\bf Correctness.}  If Kyber.Packed.CPA is $(1- \delta_\ell)$-correct and $G$ is a random oracle, then P$_\ell$-Kyber is $(1- \delta_\ell)$-correct \cite{modFOT2017}. 

\vspace{3mm}

\noindent {\bf Security.} We provide the concrete security bounds from \cite{Kyber2018}\cite{modFOT2017} which proves P$_\ell$-Kyber KEM’s CCA-security, when $G$ and $H$ are modelled as random oracles. 

\begin{lemma} [IND-CCA Secure KEM \protect{\cite[Theo. 3.2 and 3.4]{modFOT2017}}]\label{lem:cca1}
For any classical adversary \(\mathsf{A}\) that makes at most \(q_{RO}\) queries to the random oracles \(H\) and \(G\), as well as \(q_D\) queries to the decryption oracle, there exists an adversary \(\mathsf{B}\) such that
\begin{equation}
\mathsf{Adv}^{\mathsf{CCA}}_{\mathsf{P}_\ell-\mathsf{Kyber}}(\mathsf{A})  \leq 3 \mathsf{Adv}^{\mathsf{CPA}}_{\mathsf{P}_\ell-\mathsf{Kyber}}(\mathsf{B})  + q_{RO} \cdot \delta_\ell + \dfrac{3q_{RO}}{2^{256\times \ell}}. 
\end{equation}
\end{lemma}

\begin{lemma} [IND-CCA Secure KEM  \protect{\cite[Theo. 3]{Kyber2021}\cite[Theo. 4]{Kyber2018}}]\label{lem:cca2}
For any quantum adversary \(\mathsf{A}\) that makes at most \(q_{RO}\) queries to the quantum random oracles \(H\) and \(G\), as well as at most \(q_D\) (classical) queries to the decryption oracle, there exists a quantum adversary \(\mathsf{B}\) such that
\begin{equation}
\mathsf{Adv}^{\mathsf{CCA}}_{\mathsf{P}_\ell-\mathsf{Kyber}}(\mathsf{A})  \leq 8 q_{RO}^2 \cdot \delta_{\ell} + 4 q_{RO} \sqrt{(\ell+1) \cdot \mathsf{Adv}^{\mathsf{M-LWE}}_{k+\ell, k, \eta}(\mathsf{B})}.
\end{equation}
\end{lemma}

\subsection{Parameter Sets}

P$_\ell$-Kyber KEM adopts the same parameters \((q, k, \eta_1, \eta_2, d_u, d_v)\) as Kyber KEM \cite{NISTpqcfinal2024}, shown in Table~\ref{Kyber_Par}. Key and ciphertext sizes are summarized in Table~\ref{Parameter_sets}, alongside Kyber KEM parameters for encapsulating \(\ell\) messages. P$_\ell$-Kyber achieves smaller sizes, especially for large \(\ell\). Its computational cost is also lower, as \(\mathbf{u}\) (or \(\mathbf{u}'\)) is computed only once in Algorithms~\ref{alg:kyberP_encaps} and~\ref{alg:kyberP_decaps}.

\vspace{-0mm}
\begin{table}[th]
\centering
\caption{Sizes (in bytes) of keys and ciphertexts: P$_\ell$-Kyber KEM vs. Kyber KEM}
\label{Parameter_sets}\centering
\vspace{-3mm}
\begin{threeparttable}
\begin{tabular}{|c||c|c|c|c|}
\hline 
 & $m$  & $pk$ & $sk$ & $c$  \\ \hline
P$_\ell$-Kyber KEM&  $n\ell/8$  & $12kn\ell/8+32 \footnotemark[1]$ & $24kn\ell/8+32\ell+32\footnotemark[1]$  & $d_ukn/8+d_vn\ell/8
$  \\ \hline
Kyber KEM &  $n\ell/8$  & $12kn\ell/8+32\ell$ & $24kn\ell/8+64\ell$  & $d_ukn\ell/8+d_vn\ell/8
$  \\ \hline
\end{tabular}%
 \begin{tablenotes}
       \item [1] In P$_\ell$-Kyber, the random seed $\psi$ only need to be transmitted once.
     \end{tablenotes}
\vspace{-3mm}
\end{threeparttable}
\end{table}

In summary, the P$_\ell$-Kyber KEM
is a natural extension of the original Kyber KEM \cite{Kyber2021}\cite{NISTpqcfinal2024}, supporting the encapsulation of $\ell \geq 1$ secrets and a MMSE quantization $Q_{\mathsf{MMSE}}$. With $\ell=1$ and a non-MMSE quantization $Q_{\mathsf{Kyber}}$, the P$_\ell$-Kyber KEM reduces to the original Kyber KEM. The advantages of P$_\ell$-Kyber KEM are twofold: its CER approaching a constant value, $d_v$, as $\ell$ increases, while its DFR can be evaluated numerically. Note that the DFR analysis of current CER-oriented approaches \cite{LWEchannel2022}\cite{liu2023lattice}\cite{liu2024SCKyber} relies on the CLT assumption.  However, the downside of P$_\ell$-Kyber is its DFR increases linearly with the value of $\ell$. Since a high DFR will impact the security bounds in Lemmas \ref{lem:cca1} and \ref{lem:cca2},  the value of $\ell$ is bounded, e.g., $\ell \leq 16$ in P$_\ell$-KYBER512. In the next section, we will show how to reduce the DFR of the P\(_\ell\)-Kyber KEM by employing lattice codes, which enables the packing of significantly more plaintexts than the uncoded version.

\section{Lattice-Coded P$_\ell$-Kyber}\label{sec: LC}
We now propose our lattice packing approaches to further reduce the DFR of multi-layer Kyber introduced in the previous section. The additional complexities of lattice encoding techniques are provided at the end of this section too.
\subsection{Lattice Vertical Encoding and Lattice Packing}
The decoding model of uncoded P$_\ell$-Kyber can be expressed as
\begin{eqnarray}
\setlength{\arraycolsep}{3pt}
{\mathbf{Y}} &=& \left\lceil q/2\right\rfloor \cdot \mathbf{m} + \mathbf{n}_{e}, \; \text{where} \label{n_e_p}
\end{eqnarray}
\begin{align}
\setlength{\arraycolsep}{3pt}
\mathbf{n}_{e} 
&= \mathbf{E}^{T}\mathbf{r}+\mathbf{e}_{2}+\mathbf{c}_{v}-\mathbf{S}^{T}\left( \mathbf{e}%
_{1}+\mathbf{c}_{u}\right).
\end{align}
Let $\mathbf{n}_{e} =[n_{0},\ldots,n_{\ell-1}]^T \in R_q^{\ell}$ be the decoding noise, where each ring element $n_{i}$ can be further interpreted as a row vector of integer coefficients $ n_{i,j}$, i.e.,
\begin{equation*}
n_{i}=[n_{i,0},n_{i,1},\ldots,n_{i,n-1}],~~~0\leq i\leq \ell-1.
\end{equation*}
It is more convenient to represent $\mathbf{n}_e$ in matrix form:
\begin{equation}
\phi(\mathbf{n}_e) = \begin{bNiceMatrix}[c, columns-width=1mm]
n_{0,0} & n_{0,1} & \cdots & n_{0,n-1}  \\
n_{1,0} & n_{1,1} & \cdots & n_{1,n-1}  \\
\vdots & \vdots  & \cdots & \vdots  \\
n_{\ell-1,0} & n_{\ell-1,1} & \cdots & n_{\ell-1,n-1}\end{bNiceMatrix}.   \label{ne_M}
\end{equation}
where $\phi$ is given in Definition \ref{def:ME}. We can also represent $\mathbf{m}$ in matrix form:
\begin{equation}
\phi(\mathbf{m}) = \begin{bNiceMatrix}[c, columns-width=1mm]
m_{0,0} & m_{0,1} & \cdots & m_{0,n-1}  \\
m_{1,0} & m_{1,1} & \cdots & m_{1,n-1}  \\
\vdots & \vdots  & \cdots & \vdots  \\
m_{\ell-1,0} & m_{\ell-1,1} & \cdots & m_{\ell-1,n-1}
\end{bNiceMatrix}, \label{m_M}
\end{equation}
where $\mathbf{m} =[m_{0},\ldots,m_{\ell-1}]^T \in R_q^{\ell}$ is the $\ell$ messages, and $m_i$ can be further interpreted as a row vector of coefficients $ m_{i,j} \in \mathbb{Z}_2$, i.e., $m_{i}=[m_{i,0},m_{i,1},\ldots,m_{i,n-1}]$.

By substituting \eqref{ne_M} and \eqref{m_M} to \eqref{n_e_p}, the decoding model $\mathbf{Y}$ can be conveniently expressed in matrix form $\phi(\mathbf{Y})$:
\vspace{-1mm}
\begin{eqnarray}
\setlength{\arraycolsep}{3pt}
&&\left\lceil \dfrac{q}{2}\right\rfloor \cdot \underbrace{\begin{bNiceMatrix}[c, columns-width=0mm,first-row]
 &\text{\scriptsize $\in \{0,1\}^\ell$} & &    \\
m_{0,0} & \tikzmarknode{M00}{m_{0,1}} & \cdots & m_{0,n-1}  \\
m_{1,0} & m_{1,1} & \cdots & m_{1,n-1}  \\
\vdots & \vdots  & \cdots & \vdots  \\
m_{\ell-1,0} &  \tikzmarknode{MLL}{m_{\ell-1,1}} & \cdots & m_{\ell-1,n-1}
\end{bNiceMatrix}}_{\phi(\mathbf{m})}+ \underbrace{\begin{bNiceMatrix}[c, columns-width=0mm,first-row]
 &\text{\scriptsize i.i.d. RVs} & &    \\
{n_{0,0}} & \tikzmarknode{N00}{n_{0,1}} & \cdots & n_{0,n-1}  \\
\tikzmarknode{N10}{n_{1,0}} & n_{1,1} & \cdots & \tikzmarknode{N1N}{n_{1,n-1}}  \\
\vdots & \vdots  & \cdots & \vdots  &\\
{n_{\ell-1,0}} & \tikzmarknode{NLL}{n_{\ell-1,1}} & \cdots & n_{\ell-1,n-1} 
\end{bNiceMatrix}}_{\phi(\mathbf{n}_e)}  \begin{NiceMatrix}[c]
 \\
\text{\scriptsize depend.}\\
\text{\scriptsize RVs}\\
\\
\\
\end{NiceMatrix} \notag \\ \label{kyberp_decoding_model}
\end{eqnarray}

\vspace{-0mm}

\begin{tikzpicture}[remember picture,overlay]
 \draw let \p1=($(NLL)-(N00)$),\n1={atan2(\y1,\x1)} in 
 node[rotate fit=\n1,fit=(N00) (NLL),draw,rounded corners,inner sep=2pt]{};
  \draw let \p1=($(N1N)-(N10)$),\n1={atan2(\y1,\x1)} in 
 node[rotate fit=\n1,fit=(N10) (N1N),draw,rounded corners,inner sep=2pt]{};
   \draw let \p1=($(MLL)-(M00)$),\n1={atan2(\y1,\x1)} in 
 node[rotate fit=\n1,fit=(M00) (MLL),draw,rounded corners,inner sep=2pt]{};
\end{tikzpicture}

\vspace{-5mm}

\begin{lemma}[Vertical Decoding Noise]
In \eqref{kyberp_decoding_model}, the entries within each column of $\mathbf{n}_e$ are independent and identically distributed (i.i.d.) random variables.
\end{lemma}

\begin{proof}
Recalling that $\mathbf{n}_{e} 
= \mathbf{E}^{T}\mathbf{r}+\mathbf{e}_{2}+\mathbf{c}_{v}-\mathbf{S}^{T}\left( \mathbf{e}%
_{1}+\mathbf{c}_{u}\right).$
Without loss of generality, we study the distribution of the first column in $\phi(\mathbf{n}_e)$, i.e., $\Pr(n_{0,0},\ldots, n_{\ell-1,0})$. For $0\leq i \leq \ell-1$, we observe that $n_{i,0}$ is generated by the same realization of  $(\mathbf{r}, \mathbf{e}
_{1}+\mathbf{c}_{u})$, denoted as $(\mathbf{a},\mathbf{b})$. We can interpret $n_{i,0}$ as a deterministic function of random variables with fixed parameters $(\mathbf{a},\mathbf{b})$:
\begin{eqnarray}
 n_{i,0}=g_{(\mathbf{a},\mathbf{b})}(E_i, \mathbf{s}_i, e_{2,i}, c_{v,i}),   
\end{eqnarray}
where $E_i$ and $\mathbf{s}_i$ are the $i$-th columns in $\mathbf{E}$ and $\mathbf{S}$, receptively. And $e_{2,i}$ and $c_{v,i}$ represent the $i$-th elements in $\mathbf{e}_2$ and $\mathbf{c}_v$, respectively.  Since $(E_i, \mathbf{s}_i, e_{2,i}, c_{v,i})$ are mutually independent for $0 \leq i \leq \ell-1$, $\{n_{i,0}\}_{i=0}^{\ell-1}$ are i.i.d. random variables.  
\end{proof}

\begin{remark}
 The current encoding schemes for (R/M-)LWE can be viewed as \emph{Horizontal Encoding}(H-Enc) \cite{LWEchannel2022}\cite{liu2023lattice}\cite{FrodoCong2022}\cite{NewhopeECC2018}, where the rows of $\phi(\mathbf{m})$ are encoded. The major issue of H-Enc is that the elements in each row of $\phi(\mathbf{n}_e)$ are dependent. The DFR analysis has to assume that the noise coefficients in each row are mutually independent (CLT), which may result in an underestimated DFR. Given that each column of \(\phi(\mathbf{n}_e)\) consists of i.i.d. RVs, it is natural to encode \(\phi(\mathbf{m})\) column-wise, thereby circumventing reliance on the CLT assumption.
\end{remark}

\begin{definition}[Lattice-Based Vertical Encoding (LV-Enc)]
Given $\mathbf{m}\leftarrow \mathcal{M}_{p, \ell}^n$, the lattice encoder and decoder are defined by
\begin{eqnarray}
\mathsf{Enc_\Lambda}&:& \mathcal{M}_{p, \ell}^n \rightarrow (\mathcal{L}(\hat{\mathbf{B}})\cap \mathbb{Z}_p^{\ell})^n \notag \\
\mathsf{Dec_\Lambda}&:& (\mathcal{L}(\hat{\mathbf{B}})\cap \mathbb{Z}_p^{\ell})^n \rightarrow \mathcal{M}_{p, \ell}^n
\end{eqnarray}
where $\mathsf{Enc_\Lambda}(\phi(\mathbf{m})):=\phi(\hat{\mathbf{m}}) = \hat{\mathbf{B}} \phi(\mathbf{m}) \bmod p$ encodes the $n$ columns of $\phi(\mathbf{m})$ into $n$ lattice points $\phi(\hat{\mathbf{m}})$ in a column-wise manner, and $\mathsf{Dec_\Lambda}(\phi(\hat{\mathbf{m}})): = \phi({\mathbf{m}}) = \mathsf{CVP_{HS}}( \phi(\hat{\mathbf{m}}))$ takes the input of $\phi(\hat{\mathbf{m}})$ and returns $\phi(\mathbf{m})$. The notations of the lattice $\mathcal{L}(\hat{\mathbf{B}})$, the matrix $\hat{\mathbf{B}} \in \mathbb{Z}^{\ell \times \ell}$, the message space $\mathcal{M}_{p, \ell}$, the decoder $\mathsf{CVP_{HS}}(\cdot)$, and the hypercube shaping $\mathcal{L}(\hat{\mathbf{B}})\cap \mathbb{Z}_p^{\ell}$ are given in Section \ref{Sec:Lattice_def}.
\end{definition}

To gain more insight, the coded version of \eqref{kyberp_decoding_model} is described by $\phi(\mathbf{Y})=$
\vspace{-3mm}
\begin{eqnarray}
\setlength{\arraycolsep}{3pt}
&&\left\lceil \dfrac{q}{p}\right\rfloor \cdot \underbrace{\hat{\mathbf{B}} \cdot \begin{bNiceMatrix}[c, columns-width=0mm,first-row]
 &\text{\scriptsize $\in \mathcal{M}_{p, \ell}$} & &    \\
m_{0,0} & \tikzmarknode{M00}{m_{0,1}} & \cdots & m_{0,n-1}  \\
m_{1,0} & m_{1,1} & \cdots & m_{1,n-1}  \\
\vdots & \vdots  & \cdots & \vdots  \\
m_{\ell-1,0} &  \tikzmarknode{MLL}{m_{\ell-1,1}} & \cdots & m_{\ell-1,n-1}
\end{bNiceMatrix} \bmod p}_{\phi(\hat{\mathbf{m}})} + \underbrace{\begin{bNiceMatrix}[c, columns-width=0mm,first-row]
 &\text{\scriptsize i.i.d. RVs} & &    \\
{n_{0,0}} & \tikzmarknode{N00}{n_{0,1}} & \cdots & n_{0,n-1}  \\
\tikzmarknode{N10}{n_{1,0}} & n_{1,1} & \cdots & \tikzmarknode{N1N}{n_{1,n-1}}  \\
\vdots & \vdots  & \cdots & \vdots  &\\
{n_{\ell-1,0}} & \tikzmarknode{NLL}{n_{\ell-1,1}} & \cdots & n_{\ell-1,n-1} 
\end{bNiceMatrix}}_{\phi(\mathbf{n}_e)}  \notag \\ \label{kyberp_decoding_C_model}
\end{eqnarray}

\vspace{-0mm}

\begin{tikzpicture}[remember picture,overlay]
 \draw let \p1=($(NLL)-(N00)$),\n1={atan2(\y1,\x1)} in 
 node[rotate fit=\n1,fit=(N00) (NLL),draw,rounded corners,inner sep=2pt]{};
  \draw let \p1=($(MLL)-(M00)$),\n1={atan2(\y1,\x1)} in 
 node[rotate fit=\n1,fit=(M00) (MLL),draw,rounded corners,inner sep=2pt]{};
\end{tikzpicture}

\vspace{-3mm}

\noindent {\bf Lattice packing.} The distribution of the noise vectors in LV-Enc is bounded by a hypersphere with high probability (We will show this in Lemma \ref{lem:KyberP_DFR}). Since the addition in (\ref{kyberp_decoding_model}) is over the modulo $q$ domain, the LV-Enc problem in P$_\ell$-Kyber can be viewed as a \emph{lattice packing} problem: an arrangement of non-overlapping spheres within a hypercube $\mathbb{Z}_{q}^{\ell}$. The model in \eqref{kyberp_decoding_model} uses the integer lattice codes $\lfloor q/2\rceil\mathbb{Z}_{2}^{\ell}$ for packing purposes, which is far from optimal. Even for very small dimensions $\ell$,
there exists much denser lattice packings than cubic ones.

\begin{definition}[Coded P$_\ell$-Kyber PKE]
The encryption and decryption of the uncode P$_\ell$-Kyber PKE can be easily adapted for the coded version by implementing the following modifications.
\begin{itemize}
    \item Coded version of Algorithm \ref{alg:kyberP_enc}
    \begin{itemize}
        \item input message space: replace $ \{0,1\}^{\ell \times n} $ by $\mathcal{M}_{p, \ell}^n$
        \item Step 5: replace $ \mathbf{m} $ by $ \phi^{-1}(\mathsf{Enc_\Lambda}(\phi(\mathbf{m})))$
    \end{itemize}
     \item Coded version of Algorithm \ref{alg:kyberP_dec}
     \begin{itemize}
        \item Step 3: replace $\mathsf{Compress}_{q}(\mathbf{v}-\mathbf{S}^{T}\bu,1)$ by $ \mathsf{Dec_\Lambda}(\phi(\mathbf{v}-\mathbf{S}^{T}\bu))$
    \end{itemize}
\end{itemize}
\end{definition}

For the choice of $\mathcal{L}(\hat{\mathbf{B}})$, in this work, we consider E$8$ lattice with  $\ell=8$, Barnes–Wall lattice with $\ell=16$ (BW16),  and Leech lattice with $\ell=24$ (Leech24)\cite{BK:Conway93}. These lattices provides the best known sphere packing in their dimension $\ell$. Since the coefficients in Kyber are integers, we will scale the original generator matrix to an integer matrix and utilize the corresponding $\hat{\mathbf{B}}$.

\vspace{3mm}

\noindent {\bf Correctness.} Let $\lambda (p)$ be the length of a shortest non-zero vector in the lattice $\mathcal{L}(\lfloor
q/p\rceil\hat{\mathbf{B}})$. The correct decoding radius of HS-CVP decoder (i.e., packing radius of $\mathcal{L}(\lfloor
q/p\rceil\hat{\mathbf{B}}$) is $\lambda (p)/2$. We show below the correctness of coded P$_\ell$-Kyber PKE.

\begin{lemma}[DFR of Coded P$_\ell$-Kyber PKE]\label{lem:KyberP_DFR}
\begin{equation}
\delta_{\ell} \leq n\exp \left(-\theta \lambda (p)^2/4 + \ell\log(M_{n^2_{0,0}}(\theta)) \right),\label{KyberP_DFR_LE}
\end{equation}
where $M_{X}(\theta)$ is the moment generating function of $X$, defined in Section \ref{sec:nd}.
\end{lemma}

\begin{proof}
We first study the DFR for the first lattice point (the first column in $\hat{\mathbf{m}}$).
\begin{equation}
\delta^{(1)}=\Pr\left (\sum_{i=0}^{\ell-1} n_{i,0}^2 \geq \lambda (p)^2/4\right)
\end{equation}
Since $n_{i,0}$, for $0 \leq i \leq \ell-1$ are i.i.d., we have 
\begin{equation}
M_{\sum_{i=0}^{\ell-1} n_{i,0}^2}(\theta) = M_{ n_{0,0}^2}(\theta)^{\ell}
\end{equation}
Using Chernoff bound and union bound, we obtain \eqref{KyberP_DFR_LE}.
\end{proof}

\begin{remark}
We numerically search the optimal $\theta$ which satisfies
\begin{equation}
    \theta= \arg \min_{\theta'\in \mathbb{R}} \exp\left(-\theta'\lambda (p)^2/4 + \ell\log(M_{n^2_{0,0}}(\theta')) \right).
\end{equation}
The distribution of $n^2_{0,0}$ can be obtained from the Python code in \cite{KyberCode}. For demonstration purposes, we plot the distribution of $n^2_{0,0}$ for P$_\ell$-KYBER1024 in Fig. \ref{K1024_N_CDF_Plot}. For different choices of $\mathcal{L}(\hat{\mathbf{B}})$, the values of $\lambda(p)$ are listed in Table \ref{KyberP_LC}.
\end{remark}

\noindent {\bf Plaintext size and CER.} Let \( K_{p,\ell} = \log_2(|\mathcal{M}_{p,\ell}|) \) denote the information bit length per lattice codeword in \eqref{kyberp_decoding_C_model}. According to Definition \ref{def:HS}, the plaintext size of coded P$_\ell$-Kyber, $N$ (in bits), can be computed by 
\begin{equation}
N=n \cdot  K_{p,\ell} =n \cdot {\textstyle\sum\nolimits}_{i=1}^{\ell}\log_{2}(p/\pi_{i}), \label{b_pl}
\end{equation}%
where $\pi_{i}$ is given in \eqref{HC_shaping}, for $i=1,\ldots, \ell$. The CER of coded P$_\ell$-Kyber is
\begin{equation}
    \rho_\ell = \dfrac{knd_u+\ell n d_v}{N} = \dfrac{kd_u+\ell d_v}{K_{p,\ell}}.
\end{equation}

Table~\ref{KyberP_LC} lists the $(\delta_{\ell},\rho_\ell)$ values for various lattice encoders. In comparison to the values of $\rho_\ell$ for uncoded P$_\ell$-Kyber, we notice that the coded version has a smaller $\rho_\ell$. This can be explained by $K_{p,\ell} \geq \ell$, i.e., the uncoded P$_\ell$-Kyber embeds $\ell$ secret bits in each column of $\phi(\mathbf{m})$ in \eqref{kyberp_decoding_model}, while the coded version encodes $K_{p,\ell}$ secret bits in each column of $\phi(\hat{\mathbf{m}})$ in \eqref{kyberp_decoding_C_model}. Coded P$_{24}$-KYBER1024 reduces CER by $90\%$ and DFR by \(2^{107}\), relative to KYBER1024.

\vspace{3mm}

\noindent {\bf Security.} The security proofs of coded P$_\ell$-Kyber PKE and KEM are the same as the uncoded versions and thus omitted.

\begin{figure}[tbp]
\centering
\includegraphics[width=0.8\textwidth]{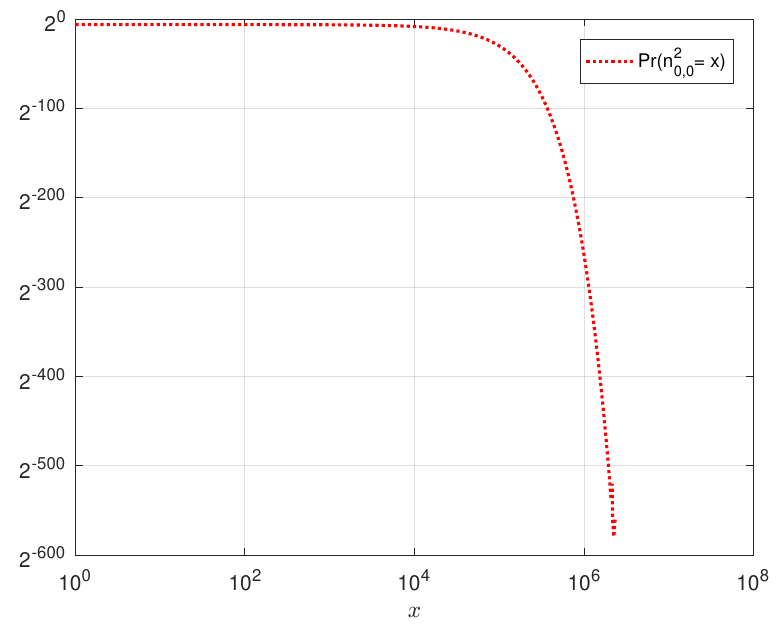} \vspace{-3mm} 
\caption{P$_\ell$-KYBER1024: distribution of $n_{0,0}^2$ from the Python code in \cite{KyberCode}}
\label{K1024_N_CDF_Plot}
\vspace{-3mm}
\end{figure}

\vspace{-0mm}
\begin{table}[tbh]
\centering
\caption{Lattice Codes for P$_\ell$-Kyber: using Kyber's $(k, q, \eta_1, \eta_2, d_u,d_v)$ }
\vspace{-3mm}
\label{KyberP_LC}\centering
\begin{tabular}{|c||c||c|c|c|}
\hline
& Uncoded & \multicolumn{3}{c|}{Coded} \\ \hline
Lattice  & $\mathbb{Z}^{8}$ & $\text{E8}$ & $\text{BW16}$ & $\text{Leech24}$ \\ \hline
$\ell$ & $8$ & $8$  & $16$ & $24$ \\ \hline
$p$ &  $2$  &  $4$ &  $4$ & $8$ \\ \hline
$\lambda(p)/(2\left\lfloor q/2\right\rceil )$  & $0.5$ & $0.7067$ & $0.7067$ & $0.7067$ \\ \hline
$K_{p,\ell} $ & $8$ & $8$ &$20$ & $36$ \\ \hline
$N$ (in bits) & \makecell{$8n$ \\ ($8$ AES keys)} & \makecell{$8n$ \\ ($8$ AES keys)} & \makecell{$20n$ \\ ($20$ AES keys)}  & \makecell{$36n$ \\ ($36$ AES keys)}\\ \hline
P$_\ell$-KYBER512 & \makecell{ $\delta_\ell=2^{-139}$  \\  $\rho_\ell=6.5$}  & \makecell{ $\delta_\ell=2^{-225}$ \\  $\theta=1.4\times 10^{-4}$ \\  $\rho_\ell=6.5$} & \makecell{ $\delta_\ell=2^{-184}$ \\  $\theta=1.3\times 10^{-4}$ \\  $\rho_\ell=4.2$} & \makecell{ $\delta_\ell=2^{-155}$ \\  $\theta=1.2\times 10^{-4}$ \\  $\rho_\ell=3.2$} \\ \hline
P$_\ell$-KYBER768 &  \makecell{ $\delta_\ell=2^{-166}$  \\  $\rho_\ell=7.8$} & \makecell{ $\delta_\ell=2^{-267}$ \\  $\theta=1.7\times 10^{-4}$ \\  $\rho_\ell=7.8$} &  \makecell{ $\delta_\ell=2^{-217}$ \\  $\theta=1.53\times 10^{-4}$ \\  $\rho_\ell=4.7$} & \makecell{ $\delta_\ell=2^{-183}$ \\  $\theta=1.4\times 10^{-4}$ \\  $\rho_\ell=3.5$} \\ \hline
P$_\ell$-KYBER1024 & \makecell{ $\delta_\ell=2^{-187}$  \\  $\rho_\ell=10.5$}  & \makecell{ $\delta_\ell=2^{-336}$ \\  $\theta=1.88\times 10^{-4}$ \\  $\rho_\ell=10.5$} & \makecell{ $\delta_\ell=2^{-306}$ \\  $\theta=1.85\times 10^{-4}$ \\  $\rho_\ell=6.2$}  & \makecell{ $\delta_\ell=2^{-281}$ \\  $\theta=1.79\times 10^{-4}$ \\  $\rho_\ell=4.6$} \\ \hline
\end{tabular}%
\vspace{-3mm}
\end{table}

\subsection{Side-Channel Attack, Constant-Time Decoder, and Complexity}
The implementation of the lattice decoder may be susceptible to side-channel attacks. In \cite{ECCTimingAttack2019}, The authors note that the decoding process typically recovers valid codewords more quickly than those containing errors. This timing information can be exploited to differentiate between valid ciphertexts and failing ciphertexts. However, this attack can be mitigated by employing a constant-time decoder. The fundamental idea is to partition the lattice $\mathsf{\Lambda}$  into the cosets of a specific sublattice $\mathsf{\Lambda}'$. The constant-time decoding problem for $\mathsf{\Lambda}$ can be reduced to the constant-time decoding problem for $\mathsf{\Lambda}'$. by exhaustively searching through all cosets of $\mathsf{\Lambda}'$.

In Table \ref{LD_TimeComplexity}, we recall the time complexity of existing constant-time lattice decoders in \cite{liu2023lattice}. We count the total numbers of additional-equivalent operations as in \cite{SloaneLeech86}. Let $\mathcal{L}(\mathbf{D}_{\ell})$ be the $\ell$-dimensional checkerboard lattice \cite{BK:Conway93}, and $\mathcal{L}(\mathbf{Q}_{24})$ be the Leech quarter lattice \cite{VD93Leech}. From an engineering perspective, the energy consumption associated with communication significantly exceeds that of computation. We believe that lattice codes can help reduce the overall energy consumption.

\begin{table}[th]
\centering
\caption{Constant-time lattice decoders: time complexity}
\label{LD_TimeComplexity}\centering
\vspace{-3mm}
\begin{tabular}{|c|c|c|c|c|c|}
\hline
Lattice decoder & $\mathbb{Z}$ \cite{Kyber2021}  & $\text{E8}$ \cite{SloaneDn1982} \cite{BK:Conway93} & $\text{BW16}$ \cite{BK:Conway93}\cite{FrodoCong2022}& $\text{Leech24}$ \cite{Conway1984}\cite{SloaneLeech86} & $\text{Leech24}$ \cite{constanttimeLeech2016}\\ \hline
Lattice dimension & $1$ & $8$ & $16$ & $24$ & $24$ \\ \hline
$\mathsf{\Lambda}'$& $\mathbb{Z}$ & $\mathcal{L}(\mathbf{D}_{8})$ & $\mathcal{L}(\mathbf{D}_{16})$ & $\mathcal{L}(\mathbf{D}_{24})$ & $\mathcal{L}(\mathbf{Q}_{24})$ \\ \hline
$\#$ of operations & $1$ & $64$ & $2048$ & $786432$ & $\approx 3974$ \\ \hline
\end{tabular}
\vspace{-3mm}
\end{table}

In summary, we demonstrate that lattice-based vertical encoding via ciphertext packing effectively reduces both DFR and CER of the orginal Kyber, by increasing the plaintext size $N$, without relying on the CLT assumption for decoding noise. 
\section{Truncated Lattice-Coded P$_\ell$-Kyber: $N=256$ bits} \label{sec: TLC}
In practical scenarios, the plaintext size is commonly fixed ($\ell=1$ or $\ell=2$), e.g., \( N = 256 \) bits as opposed to having large $\ell>2$ in Section 4. In this section, we will show that plaintext and ciphertext size of coded P\(_\ell\)-Kyber can be naturally adapted by truncation.

\subsection{Coded P$_\ell$-Kyber with Truncated Ciphertext}
Let us start by providing the definition of truncation. \begin{definition}[Truncation Function] Let \( \mathbf{A} \in \mathbb{F}^{\ell \times n} \) be a matrix over a field \( \mathbb{F} \), where \( A = [a_{i,j}] \) and each entry \( a_{i,j} \in \mathbb{F} \) for \( 0 \leq i \leq \ell-1 \), \( 0 \leq j \leq n-1 \). Let \( t \in \mathbb{N} \) be such that \( 1 \leq t \leq n \). The truncation function \( \operatorname{Trunc}_t \) is defined as:
\begin{equation*}
\mathsf{Trunc}_{t}(\mathbf{A} ) := \left[ a_{i,j} \right]_{0 \leq i \leq \ell-1,\, 0 \leq j \leq t-1}.    
\end{equation*}
That is, \( \mathsf{Trunc}_t( \mathbf{A}) \) returns a matrix consisting of the first \( t \) columns of \(  \mathbf{A} \), by removing its last \( n-t \) columns.
\end{definition}

The plaintext and ciphertext size of coded/uncoded P$_\ell$-Kyber can be easily adjusted according to a given plaintext size, e.g., $N=256$ bits. The basic idea is to truncate the ciphertext $\mathbf{v}= Q_{\mathsf{MMSE}, 2^{d_{v}}}(\boldsymbol{\mathrm{T}}^{T}\boldsymbol{\mathrm{r}}+\mathbf{e}_2+\left\lceil {q}/{2}\right\rfloor \cdot \mathsf{Enc_\Lambda}(\mathbf{m}))$. Let $\mathbf{v} =[v_{0},\ldots,v_{\ell-1}]^T \in R_q^{\ell}$, where each ring element $v_{i}$ can be further interpreted as a row vector of integer coefficients $ v_{i,j}$, i.e.,
\begin{equation*}
v_{i}=[v_{i,0},n_{i,1},\ldots,v_{i,n-1}],~~~0\leq i\leq \ell-1.
\end{equation*}
The vector \( \mathbf{v} \) can be equivalently expressed in matrix form:
\begin{equation}
\phi(\mathbf{v}) = \begin{bNiceMatrix}[c, columns-width=1mm]
v_{0,0}& v_{0,1} & \cdots & v_{0,n-1} \\
v_{1,0} & v_{1,1} & \cdots & v_{1,n-1} \\
\vdots & \vdots  & \cdots & \vdots  \\
v_{\ell-1,0}& v_{\ell-1,1} & \cdots & v_{\ell-1,n-1}\end{bNiceMatrix}.   \label{v_M}
\end{equation}

Due to P$_{\ell}$-Kyber’s column-wise decoding structure, i.e., $\mathsf{Dec_\Lambda}(\phi(\mathbf{v}-\mathbf{S}^{T}\bu))$, the last $n - t$ columns of $\phi(\mathbf{v})$ can be removed, resulting in a truncated vector $\hat{\mathbf{v}}$, whose matrix representation is given by:
\begin{equation}
  \phi(\hat{\mathbf{v}})= \mathsf{Trunc}_t(\phi(\mathbf{v}))=\begin{bNiceMatrix}[c, columns-width=1mm]
v_{0,0} & v_{0,1}  & \cdots & v_{0,t-1}  \\
v_{1,0} & v_{1,1} & \cdots & v_{1,t-1}  \\
\vdots   & \cdots & \vdots  \\
v_{\ell-1,0} & v_{\ell-1,1} & \cdots & v_{\ell-1,t-1}\end{bNiceMatrix}.
\end{equation}
The corresponding decoding model, i.e., $\phi({\mathbf{Y}})=\mathsf{Trunc}_t(\phi(\mathbf{v}-\mathbf{S}^{T}\bu))$, is given by 
\begin{eqnarray}
\setlength{\arraycolsep}{3pt}
&&\left\lceil \dfrac{q}{p}\right\rfloor \cdot \underbrace{\hat{\mathbf{B}} \cdot \begin{bNiceMatrix}[c, columns-width=0mm,first-row]
 &\text{\scriptsize $\in \mathcal{M}_{p, \ell}$} & &    \\
m_{0,0} & \tikzmarknode{M00}{m_{0,1}} & \cdots & m_{0,t-1}  \\
m_{1,0} & m_{1,1} & \cdots & m_{1,t-1}  \\
\vdots & \vdots  & \cdots & \vdots  \\
m_{\ell-1,0} &  \tikzmarknode{MLL}{m_{\ell-1,1}} & \cdots & m_{\ell-1,t-1}
\end{bNiceMatrix} \bmod p}_{\phi(\hat{\mathbf{m}})} + \underbrace{\begin{bNiceMatrix}[c, columns-width=0mm,first-row]
 &\text{\scriptsize i.i.d. RVs} & &    \\
{n_{0,0}} & \tikzmarknode{N00}{n_{0,1}} & \cdots & n_{0,t-1}  \\
\tikzmarknode{N10}{n_{1,0}} & n_{1,1} & \cdots & \tikzmarknode{N1N}{n_{1,t-1}}  \\
\vdots & \vdots  & \cdots & \vdots  &\\
{n_{\ell-1,0}} & \tikzmarknode{NLL}{n_{\ell-1,1}} & \cdots & n_{\ell-1,t-1} 
\end{bNiceMatrix}}_{\phi(\mathbf{n}_e)}  \notag \\ \label{kyberp_decoding_TC_model}
\end{eqnarray}

\vspace{-2mm}

\begin{tikzpicture}[remember picture,overlay]
 \draw let \p1=($(NLL)-(N00)$),\n1={atan2(\y1,\x1)} in 
 node[rotate fit=\n1,fit=(N00) (NLL),draw,rounded corners,inner sep=2pt]{};
  \draw let \p1=($(MLL)-(M00)$),\n1={atan2(\y1,\x1)} in 
 node[rotate fit=\n1,fit=(M00) (MLL),draw,rounded corners,inner sep=2pt]{};
\end{tikzpicture}

\noindent The plaintext size is reduced to $N=t K_{p,\ell}$, where $K_{p,\ell}$ is given in Table \ref{KyberP_LC}. The time complexity of truncation is $O(\ell(n-t))$.

\begin{definition}[P$_{t,\ell}$-Kyber PKE]
The encryption and decryption of the uncode P$_\ell$-Kyber PKE can be easily adapted for the truncated coded version, denoted by P$_{t,\ell}$-Kyber, by implementing the following modifications.
\begin{itemize}
    \item Truncated coded version of Algorithm \ref{alg:kyberP_enc}
    \begin{itemize}
        \item input message space: replace $ \{0,1\}^{256} $ by $\mathcal{M}_{p, \ell}^{t}$
        \item Step 5: $\mathbf{v}\coloneqq Q_{\mathsf{MMSE}, 2^{d_{v}}}(\phi^{-1}(\mathsf{Trunc}_{t} (\phi(\boldsymbol{\mathrm{T}}^{T}\boldsymbol{\mathrm{r}}+\mathbf{e}_2))+\left\lceil {q}/{2}\right\rfloor \cdot \mathsf{Enc_\Lambda}(\phi(\mathbf{m}))))$
    \end{itemize}
     \item Truncated coded version of Algorithm \ref{alg:kyberP_dec}
     \begin{itemize}
         \item Step 3:  replace $\mathsf{Compress}_{q}(\mathbf{v}-\mathbf{S}^{T}\bu,1)$ by $ \mathsf{Dec_\Lambda}(\phi(\mathbf{v})-\mathsf{Trunc}_{t}(\phi(\mathbf{S}^{T}\bu)))$
    \end{itemize}
\end{itemize}
\end{definition}

\vspace{3mm}

\noindent {\bf Correctness.} Using Lemma \ref{lem:KyberP_DFR} with $n=t$, the DFR of P$_{t,\ell}$-Kyber is given by
\begin{equation}
\delta_{t,\ell} \leq t\exp \left(-\theta \lambda (p)^2/4 + \ell\log(M_{n^2_{0,0}}(\theta)) \right),\label{KyberTP_DFR_LE}
\end{equation}
where $M_{X}(\theta)$ is the moment generating function of $X$, defined in Section \ref{sec:nd}.

\vspace{3mm}

\noindent {\bf Plaintext and ciphertext size.} For a fixed plaintext size \( N \) (in bits), e.g., \( N = 256 \), one can select the number of packed codewords as \( t  = N / K_{p,\ell}\), where \( K_{p,\ell} = \log_2(|\mathcal{M}_{p,\ell}|) \) is the information bit length per lattice codeword in Table \ref{KyberP_LC}. The resulting ciphertext size \( M \), corresponding to the pair \( (\mathbf{u}, \mathbf{v}) \), is given by:
\begin{equation}
    M = knd_u+ t K_{p, \ell} d_v= knd_u+Nd_v. 
\end{equation}
Given \( N = 256 \) bits and same $(k,d_u,d_v)$, the ciphertext size of P$_{t,\ell}$-Kyber is the same as that of standard Kyber. The CER of P$_{t,\ell}$-Kyber is given by
\begin{equation}
    \rho_{t,\ell}=\dfrac{M}{N}=\dfrac{knd_u+ tK_{p, \ell} d_v}{N}.
\end{equation}

\vspace{3mm}

\noindent {\bf Security.} The ciphertext of P\(_{t,\ell}\)-Kyber is derived by deterministically discarding \((n - t)\ell\) coefficients from the coded P\(_\ell\)-Kyber ciphertext, thereby preserving at least the same level of security as coded P\(_\ell\)-Kyber.

\subsection{CER Reduction Through Tighter Compression Parameters}

Since \( \delta_{t,\ell} \leq \delta_{\ell} \), and the values of \( \delta_{\ell} \) are significantly lower than that of standard Kyber (see Table \ref{KyberP_LC}), the CER can be reduced by selecting smaller compression parameters \( (d_u, d_v) \).

In Table \ref{TC_P_Kyber_Par}, we evaluate the performance of P\(_{t,\ell}\)-Kyber with parameters \( (N=32 \; \text{bytes}, t = 32, \ell = 8) \), employing the \text{E8} lattice encoder. The table reports DFR and CER values for various \( (d_u, d_v) \) configurations. Relative to the original KYBER1024, coded P\(_{t,\ell}\)-KYBER1024 achieves a $10.2\%$ reduction in CER and a DFR reduction by a factor of $2^{30}$, using \( (d_u = 10, d_v = 4) \). If a DFR of \( 2^{-128} \) is deemed sufficient, the CER can be further reduced by $16.3\%$ with \( (d_u = 9, d_v = 5) \).

\begin{table}[ht]
\caption{Parameters of P$_{t,\ell}$-Kyber with \( ( \ell = 8, t = 32) \)}
\label{TC_P_Kyber_Par}\centering
\vspace{-3mm}
\begin{tabular}{|c||c|c|c|c|c|c|c|c|c|c|}
\hline
 &$k$ & $q$ & $\eta_{1}$ & $\eta_{2}$ & $d_{u}$ & $d_{v}$ & DFR & CER & Plaintext Size & Ciphertext Size \\ \hline
KYBER1024 \cite{NISTpqcfinal2024}& $4$ & $3329$ & $2$ & $2$ & $11$ & $5$ & $2^{-174}$ &$49$ & $32$ bytes & $1568$ bytes\\ \hline
P$_{t,\ell}$-KYBER1024&  $4$ & $3329$ & $2$ & $2$ & $10$ & $4$ & $2^{-204}$ &$44$ & $32$ bytes & $1408$ bytes\\ \hline
P$_{t,\ell}$-KYBER1024& $4$ & $3329$ & $2$ & $2$ & $9$ & $6$ & $2^{-138}$ &$42$ & $32$ bytes & $1344$ bytes \\ \hline
P$_{t,\ell}$-KYBER1024& $4$ & $3329$ & $2$ & $2$ & $9$ & $5$ & $2^{-128}$ &$41$ & $32$ bytes & $1312$ bytes \\ \hline
\end{tabular}
\vspace{-3mm}
\end{table}

For completeness, Table \ref{sum_contribution} lists the $(\delta=\delta_{t,\ell}, \rho=\rho_{t,\ell})$ values for P$_{t,\ell}$-Kyber, as shown in the column titled ``This Work-Lattice'' and the rows labeled ``$1-2$ AES keys''. Specifically, for the case of one AES key, we consider parameters $(t = 32, \ell = 8, d_u = 10,  d_v= 4)$, and for two AES keys, $(t = 64, \ell = 8, d_u = 10,  d_v= 4)$. In both configurations, the \text{E8} lattice encoder is utilized. Notably, we observe that P$_{t,\ell}$-Kyber achieves the encryption of two AES keys using a ciphertext size of $1536$ bytes, whereas KYBER1024 requires $1568$ bytes to encapsulate a single AES key. This highlights the inefficiency of the original Kyber encoding and suggests significant room for optimization.

In summary, for a fixed plaintext size of $N = 256$ bits, the proposed P$_{t,\ell}$-Kyber scheme achieves lower CER and DFR compared to the original Kyber, at the cost of an increased public key size of $12kn\ell/8 + 32$ bytes, as detailed in Table \ref{Parameter_sets}. However, since many cryptographic protocols—including Kyber—allow the public key to be pre-stored and reused across multiple encapsulations, the communication overhead introduced by the larger public key becomes negligible as the number of encapsulations grows.

\section{Application to Multi-Recipient KEM}

We consider the Multi-Recipient Key Encapsulation Mechanism (mKEM) in \cite{MR_PKE2020}, which securely
sends the same session key $\mathbf{m}$ to a group of $L$ recipients. For definitions, syntaxes, and security models of mKEM and mPKE, please refer to Appendix. The construction of an IND-CPA secure mPKE is in most cases a simple modification of an IND-CPA secure PKE to the multi-recipient setting.

The mPKE scheme based on P$_{t, \ell}$-Kyber (P$_{t, \ell}$-mPKE) is detailed in Algorithms \ref{alg:kyberMR_keygen}–\ref{alg:kyberMR_dec}. A global public matrix $\mathbf{A} \sim R_q^{k \times k} \coloneqq \mathsf{Sam}(\psi)$ is sampled and made available to both the sender and all recipients. Each recipient $i$ executes Algorithm \ref{alg:kyberMR_keygen} to generate a key pair $(pk_i = \mathbf{T}_i, sk_i = \mathbf{S}_i)$, and forwards $pk_i$ to the sender. The sender aggregates the public keys $\{pk_i = \mathbf{T}_i\}_{i=0}^{L-1}$ and employs Algorithm \ref{alg:kyberMR_enc} to encrypt a session key $\mathbf{m}$, producing a ciphertext $c$, which is then distributed to all recipients. Upon receiving $c$, each recipient $i$ applies Algorithm \ref{alg:kyberMR_dec} with their secret key $sk_i = \mathbf{S}_i$ to recover the session key $\mathbf{m}$.

\setcounter{algorithm}{0}
\vspace{-3mm}
\begin{algorithm}[H]
\caption{$\mathsf{P}_{t, \ell}-\mathsf{mPKE.CPA.KeyGen()}$: key generation at Recipient $i$}
\label{alg:kyberMR_keygen}
\begin{algorithmic}[1]

    \State
    $\psi,\sigma\leftarrow\left\{ 0,1\right\} ^{256}$

    \State
    $\bA\sim R_{q}^{k\times k}\coloneqq\mathsf{Sam}(\psi)$

    \State
    $(\mathbf{S}_i,\mathbf{E})\sim\beta_{\eta_1}^{k \times \ell}\times\beta_{\eta_1}^{k \times \ell}\coloneqq\mathsf{Sam}(\sigma)$

    \State
    $\mathbf{T}_i \coloneqq\boldsymbol{\mathrm{A}\mathrm{S}_i+\mathrm{E}}$

    \State \Return $\left(pk_i\coloneqq(\boldsymbol{\mathrm{T}_i},\psi),sk_i\coloneqq \mathbf{S}_i \right)$  

\end{algorithmic}
\end{algorithm}

\vspace{-12mm}

\begin{algorithm}[H]
\caption{$\mathsf{P}_{t, \ell}-\mathsf{mPKE.CPA.Enc}$ $(pk=(\{\boldsymbol{\mathrm{T}}_i\}_{i=0}^{L-1},\psi),\mathbf{m} \in \{0,1\} ^{\ell \times t}$)}
\label{alg:kyberMR_enc}
\begin{algorithmic}[1]

	\State
	$r \leftarrow \{0,1\}^{256}$

	\State
	$\boldsymbol{\mathrm{A}}\sim R_{q}^{k\times k}\coloneqq\mathsf{Sam}(\psi)$
	
	\State  $(\boldsymbol{\mathrm{r}},\boldsymbol{\mathrm{e}_{1}},\mathbf{e}_{2})\sim\beta_{\eta_1}^{k}\times\beta_{\eta_2}^{k}\times\beta_{\eta_2}^{\ell}\coloneqq\mathsf{Sam}(r)$
	
	\State  $\boldsymbol{\mathrm{u}}\coloneqq Q_{\mathsf{MMSE}, 2^{d_{u}}}(\buone)$

    \For{$i \gets 0$ to $L-1$}
    
	\State $\mathbf{v}_i\coloneqq Q_{\mathsf{MMSE}, 2^{d_{v}}}(\phi^{-1}(\mathsf{Trunc}_{t} (\phi(\boldsymbol{\mathrm{T}_i}^{T}\boldsymbol{\mathrm{r}}+\mathbf{e}_2))+\left\lceil {q}/{2}\right\rfloor \cdot \mathsf{Enc_\Lambda}(\phi(\mathbf{m}))))$

    \EndFor
	
	\State \Return $c\coloneqq(\mathsf{Index}_{2^{d_{u}}}(\boldsymbol{\mathrm{u}}),\mathsf{Index}_{2^{d_{v}}}(\mathbf{v}_0),\ldots,\mathsf{Index}_{2^{d_{v}}}(\mathbf{v}_{L-1}))$

\end{algorithmic}
\end{algorithm}

\vspace{-12mm}

\begin{algorithm}[H]
\caption{${\mathsf{P}_{t, \ell}-\mathsf{mPKE.CPA.Dec}}\ensuremath{(sk_i=\mathbf{S}_i,c=(\bu,\mathbf{v}_i))}$}
\label{alg:kyberMR_dec}
\begin{algorithmic}[1]

    \State
    $\bu\coloneqq \mathcal{C}_{2^{d_{u}}}(\mathsf{Index}_{2^{d_{u}}}(\boldsymbol{\mathrm{u}}))$

    \State
    $\mathbf{v}_i\coloneqq \mathcal{C}_{2^{d_{v}}}(\mathsf{Index}_{2^{d_{u}}}(\boldsymbol{\mathrm{v}}_i))$

    \State \Return $ \mathsf{Dec_\Lambda}(\phi(\mathbf{v}_i)-\mathsf{Trunc}_{t}(\phi(\mathbf{S}_i^{T}\bu)))$

\end{algorithmic}
\end{algorithm}
\vspace{-3mm}

\vspace{3mm}

\noindent {\bf Correctness.} The DFR of P$_{t,\ell}$-mPKE is same as the P$_{t,\ell}$-Kyber in Table \ref{TC_P_Kyber_Par}.

\vspace{3mm}

\noindent {\bf Compact Ratio.} We recall the notation of compact ratio (CR) defined in \cite{MR_PKE2020}:
\begin{equation}
    \mu =\dfrac{L\cdot \text{size of } \mathbf{u}+L\cdot\text{size of } \mathbf{v}_i}{\text{size of } \mathbf{u}+L\cdot\text{size of } \mathbf{v}_i} \approx 1+\dfrac{\text{size of } \mathbf{u}}{\text{size of } \mathbf{v}_i}
\end{equation}
which measures asymptotically how much more compact mPKE is compared to $L$ instances of the original PKE, for a large $L$. Table \ref{M_P_Kyber_Par} presents the values of $(\delta,\mu,N,M)$ for the proposed P$_{t,\ell}$-mPKE scheme and those reported in \cite{MR_PKE2020}. It can be observed that P$_{t,\ell}$-mPKE achieves a higher $\mu$, implying improved communication efficiency compared to the scheme in \cite{MR_PKE2020}. For a large $L$, the ciphertext size of P$_{t,\ell}$-mPKE  is about $80\%$ of the scheme in \cite{MR_PKE2020}.

\begin{table}[ht]
\caption{Parameters of P$_{t,\ell}$-Kyber mPKE with \( (\ell = 8, t = 32) \) and $L$ Recipients}
\label{M_P_Kyber_Par}\centering
\vspace{-3mm}
\begin{tabular}{|c||c|c|c|c|c|c|c|c|c|}
\hline
 &$k$ & $q$ & $\eta_{1}$ & $\eta_{2}$ & $d_{u}$ & $d_{v}$ & DFR & CR &  Ciphertext Size\\ \hline
KYBER1024-mPKE \cite{MR_PKE2020}& $4$ & $3329$ & $2$ & $2$ & $11$ & $5$ & $2^{-174}$ &$9.8$  & $1408+160L$ bytes\\ \hline
P$_{t,\ell}$-KYBER1024-mPKE&  $4$ & $3329$ & $2$ & $2$ & $10$ & $4$ & $2^{-204}$ &$11$ &  $1280+128L$ bytes\\ \hline
\end{tabular}
\vspace{-3mm}
\end{table}

\vspace{3mm}

\noindent {\bf Security.} We will prove that the
encryption scheme defined above is IND-CPA secure under the M-LWE hardness assumption.

\begin{definition}[IND-CPA and IND-CCA of mPKE \cite{MR_PKE2020}]\label{def:cpacca_mKEM}
We revisit the security notions for mPKE encryption, specifically indistinguishability under chosen-ciphertext attacks (IND-CCA) and chosen-plaintext attacks (IND-CPA). The advantage of
an adversary $\mathsf{A}$ is defined as 
\begin{equation}
\mathsf{Adv}^{\mathsf{CCA}}_{\mathsf{mPKE},L}(\mathsf{A}) = \left |\Pr \left(b=b' : \begin{array}{c}
\{\mathsf{pk}_i, \mathsf{sk}_i\}_{i \in [L]} \leftarrow \mathsf{KeyGen}();\\ (m_0,m_1, s) \leftarrow \mathsf{A}^{\mathsf{DEC}(\cdot)}(\{\mathsf{pk}_i\}_{i \in [L]});\\
b \leftarrow \{0,1\}; c^{*} \leftarrow \mathsf{Enc}(\{\mathsf{pk}_i\}_{i \in [L]}, m_b); \\
b' \leftarrow \mathsf{A}^{\mathsf{DEC}(\cdot)}(s, c^{*});
\end{array}\right) - 1/2 \right |
\end{equation}
where the decryption oracle is defined as $\mathsf{DEC}(\cdot) = \mathsf{Dec}(\mathsf{sk},\cdot)$. We also require that \( |m_0| = |m_1| \) and that in the second phase, the adversary \(\mathsf{A}\) is not permitted to query \(\mathsf{DEC}(\cdot)\) with the challenge ciphertext \(c^{*}\). The advantage \(\mathsf{Adv}^{\mathsf{CPA}}_{\mathsf{mPKE},L}(\mathsf{A})\) of an adversary \(\mathsf{A}\) is defined as \(\mathsf{Adv}^{\mathsf{CCA}}_{\mathsf{mPKE},L}(\mathsf{A})\), provided that \(\mathsf{A}\) cannot query $\mathsf{DEC}(\cdot)$.
\end{definition}

\begin{lemma}[IND-CPA Security of P$_{t,\ell}$-mPKE]
For any adversary $\mathsf{A}$, there exists an adversary $\mathsf{B}$
such that $\mathsf{Adv}^{\mathsf{CPA}}_{\mathsf{P}_{t,\ell}-\mathsf{mPKE},L}( \mathsf{A}) \leq L(\ell+1) \cdot \mathsf{Adv}^{\mathsf{M-LWE}}_{k+\ell, k, \eta}(\mathsf{B})$.
\end{lemma}
\begin{proof}
Let $\mathsf{A}$ be an adversary that is executed in the IND-CPA
security experiment which we call game $G_0$, i.e., $\mathsf{Adv}^{\mathsf{CPA}}_{\mathsf{P}_{t,\ell}-\mathsf{mPKE},L} =|\Pr(b = b' \text{ in game } 
 G_0)- 1/2|$. 

In game $G_1$, the $\ell$ column vectors in each public key $\mathbf{T}_i$ are simultaneously substituted with $\ell$ uniform random vectors. It is possible to verify that there exists an adversary $\mathsf{B}$ with the same running time as that of $\mathsf{A}$ such that 
\begin{equation}~\label{eq:G0G1l_m}
|\Pr(b = b' \text{ in game } G_0)- \Pr(b = b' \text{ in game } G_{1})| \leq L\ell \cdot \mathsf{Adv}^{\mathsf{M-LWE}}_{k, k, \eta}(\mathsf{B}). 
\end{equation}

In game $G_2$, the vectors $\mathbf{u}$ and $\mathbf{v}_i$ used in the generation of the challenge ciphertext are simultaneously substituted with uniform random vectors. Again, there exists an adversary $\mathsf{B}$ with the same running time as that of $\mathsf{A}$ with
\begin{equation}~\label{eq:Advmlwe_m}
|\Pr(b = b' \text{ in game } G_{1})- \Pr(b = b' \text{ in game } G_2)| \leq L \cdot \mathsf{Adv}^{\mathsf{M-LWE}}_{k+\ell, k, \eta}(\mathsf{B}). 
\end{equation}
Note that in game $G_2$, the value $\mathbf{v}_i$ from the challenge ciphertext is independent of bit $b$ and therefore $\Pr(b = b' \text{ in game } G_2) = 1/2$. Collecting the probabilities in \eqref{eq:G0G1l_m} and \eqref{eq:Advmlwe_m} yields the required bound.
\end{proof}

\vspace{3mm}

\noindent {\bf P$_{t,\ell}$-mKEM.} An IND-CPA secure P$_{t,\ell}$-mPKE can be converted into an IND-CCA secure P$_{t,\ell}$-mKEM via the (generalized) FO transform described in \cite{MR_PKE2020}; therefore, the transformation details are omitted.

\section{Conclusion}
\vspace{-1mm}
In this paper, we have investigated the effects of ciphertext packing on M-LWE based KEMs like Kyber. We have also demonstrated that by utilizing packed ciphertexts, the CER of Kyber can be decreased by over $90\%$, while still maintaining IND-CCA secure. However, a general challenge with ciphertext packing is that the DFR increases linearly with the number of packed ciphertexts. To address this issue, we introduced a coded version of packed Kyber that reduces the DFR to a negligible level. The DFR analysis can be verified numerically and does not rely on independent assumptions about the decoding noise entries. Our findings suggest that M-LWE based cryptosystems can be significantly enhanced in sizes and communication efficiency through advanced techniques in quantization, ciphertext packing, and coding.

\bibliographystyle{splncs04}
\bibliography{LIUBIB}


\section{Definitions, Syntaxes, and Security Models for Multi-Recipient PKE and KEM, adapted from~\cite{MR_PKE2020}} \label{sec:mPKEmKEM}
In this Appendix, we first provide definitions of mKEM and mPKE. We then provide a generic transformation from mPKE to mKEM. 
\subsection{Decomposable Multi-Recipient Public Key Encryption} \label{sec:mPKE}
\noindent
\begin{definition} [Decomposable Multi-Recipient Public Key Encryption] \label{def:mPKE}
A \emph{(single-message) decomposable multi-recipient public key encryption ($\mPKE$)} over a message space $\MsgSpace$ and ciphertext spaces $\CtSpace$ and $\CtSpaceSin$ consists of the following five algorithms $\mPKE = (\mSetup, \mGen, \allowbreak \mEnc, \mExt, \mDec):$ 
\begin{itemize}
	\item $\mSetup(1^\secpar) \rightarrow \pp:$ The setup algorithm on input the security parameter $1^\secpar$ outputs a public parameter $\pp$. 

	\item $\mGen(\pp) \rightarrow (\pk, \sk):$ The key generation algorithm on input a public parameter $\pp$ outputs a pair of public key and secret key $(\pk, \sk)$. 
	
	\item $\mEnc(\pp, \tuple{\pk_i}_{i \in [N]}, \msg; \randnpk, \rand_1, \cdots, \rand_N) \rightarrow \vct = \tuple{\ctnpk, \tuple{\hct_i}_{i \in [N]}}:$ The (decomposable) encryption algorithm running with randomness $(\randnpk, \rand_1, \cdots, \rand_N)$, splits into a pair of algorithms $(\mEncA, \mEncB):$ 
		\begin{itemize}
			\item $\mEncA( \pp; \randnpk ) \rightarrow \ctnpk:$ On input a public parameter $\pp$ and randomness $\randnpk$, it outputs a (public key \emph{\underline{I}ndependent}) ciphertext $\ctnpk$. 
			\item $\mEncB( \pp, \pk_i, \msg; \randnpk, \rand_i) \rightarrow \hct_i:$ 	On input a public parameter $\pp$, a public key $\pk_i$, a message $\msg \in \MsgSpace$, and randomness $(\randnpk, \rand_i)$, it outputs a (public key \emph{\underline{D}ependent)} ciphertext $\hct_i$. 
		\end{itemize}
	
	\item $ \mExt(i, \vct) \rightarrow \ct_i = (\ctnpk, \hct_i) \text{ or } \bot:$ The deterministic extraction algorithm on input an index $i \in \NN$ and a (multi-recipient) ciphertext $\vct \in \CtSpace$, outputs either a (single-recipient) ciphertext $\ct_i = (\ctnpk, \hct_i) \in \CtSpaceSin$ or a special symbol  $\bot_\Ext$ indicating extraction failure. 
	
	\item $\mDec( \sk, \ct_i ) \rightarrow \msg \text{ or } \bot:$ The deterministic decryption algorithm on input a secret key $\sk$ and a ciphertext $\ct_i \in \CtSpaceSin$, outputs either $\msg \in \MsgSpace$ or a special symbol $\bot \not \in \MsgSpace$. 
\end{itemize}
\end{definition}

\begin{definition} [Correctness] \label{def:mPKE_correctness} 
A $\mPKE$ is $\delta$-correct if 
\begin{equation}
\delta \ge \EE \left[ \max_{\msg \in \MsgSpace} \Pr_{\randnpk, \rand} \left[
\begin{array}{ccccc}
& \ctnpk \lruns \mEncA(\pp; \randnpk), \hct \gets \mEncB(\pp, \pk, \msg; \randnpk, \rand): \\
&  \msg \neq \mDec( \sk, (\ctnpk, \hct) ) \end{array} \right] \right], 
\end{equation}
where the expectation is also taken over $\pp \lruns \mSetup(1^\secpar)$ and $(\pk, \sk) \lruns \mGen(\pp)$.%
\end{definition}

\begin{definition} [$\gamma$-Spreadness] \label{def:spread}
Let $\mPKE$ be a decomposable multi-recipient $\PKE$ with message space $\MsgSpace$ and ciphertext spaces $\CtSpace$ and $\CtSpaceSin$.
For all $\pp \in \Setup(1^\secpar)$, and $(\pk, \sk) \in \Gen(\pp)$, define 
$$
	\gamma(\pp, \pk) := -\log_2 \left( \max_{\ct \in \CtSpaceSin, \msg\in \MsgSpace} \Pr_{\randnpk, \rand}\left[ \ct = \bigtuple{ \mEncA(\pp; \randnpk), \mEncB( \pp, \pk, \msg; \randnpk, \rand ) } \right]\right).
$$
We call $\mPKE$ $\gamma$-spread if $\EE[ \gamma(\pp, \pk) ] \ge \gamma$, where the expectation is taken over $\pp \lruns \mSetup(1^\secpar)$ and $(\pk, \sk) \lruns \mGen(\pp)$. 
\end{definition}

\begin{definition} [$\INDCPA$]
Let $\mPKE$ be a decomposable multi-recipient $\PKE$ with message space $\MsgSpace$ and ciphertext space $\CtSpace$. 
We define $\INDCPA$ by a game illustrated in \ref{fig:CPACCAgame} and say the (possibly quantum) adversary $\A =  (\A_1, \A_2)$ \emph{wins} if the game outputs~$1$. 
We define the advantage of $\A$ against $\INDCPA$ security of $\mPKE$ parameterized by $N \in \NN$ as
$$
\Adv^{\INDCPA}_{\mPKE, N} (\A)  = \left|   \Pr[ \A \text{ \emph{wins}} ] - 1/2 \right|. 
$$
\end{definition}

\begin{figure}[htb!]
\begin{minipage}[t]{0.5\textwidth}
\underline{ \textbf{GAME $\INDCPA$} }
    \begin{algorithmic}[1]
    	\State $\pp \samp \mSetup(1^\secpar)$
	\For{$i \in [N]$}
	   	\State $(\pk_i, \sk_i) \samp \mGen(\pp)$
	\EndFor
        \State $(\msg^*_0, \msg^*_1, \state) \lruns \A_1( \pp, \tuple{ \pk_i }_{i \in [N]} )$
        \State $b \samp \bin$
        \State $\vct^* \samp \mEnc( \pp, \tuple{ \pk_i }_{i \in [N]}, \msg_b^*  )$
        \State $b' \lruns \A_2( \pp, \tuple{ \pk_i }_{i \in [N]}, \vct^*, \state )$
        \State \Return $[ b = b' ]$
    \end{algorithmic}
\end{minipage}
%
%
\begin{minipage}[t]{0.5\textwidth}
\underline{ \textbf{GAME $\INDCCA$} }
    \begin{algorithmic}[1]
    	\State $\pp \samp \mSetup(1^\secpar)$
	\For{$i \in [N]$}
	   	\State $(\pk_i, \sk_i) \samp \mGen(\pp)$
	\EndFor
        \State $(\key_0^*, \vct^*) \lruns \mEncaps( \pp,  \tuple{ \pk_i }_{i \in [N]})$
        \State $\key_1^* \samp \KeySpace$
        \State $b \samp \bin$
        \State $b' \lruns \A^\OD( \pp, \tuple{ \pk_i }_{i \in [N]}, \vct^*, \key_b^*)$
        \State \Return $[ b = b' ]$
    \end{algorithmic}
\end{minipage}

\begin{center}
\begin{minipage}[t]{0.5\textwidth}
\underline{ Decapsulation Oracle~$\OD( i, \ct )$ }
    \begin{algorithmic}[1]
	\State $\ct_i^* := \mExt( i, \vct^* )$
    	\If{ $\ct = \ct_i^*$ }
		\State \Return $\bot$
	\EndIf
    	\State $ \key := \mDecaps( \sk_i, \ct ) $
	\State \Return $\key$ 
    \end{algorithmic}
\end{minipage}
\end{center}
\caption{ $\INDCPA$ of $\mPKE$ and $\INDCCA$ of $\mKEM$. }\label{fig:CPACCAgame}
\end{figure}

\subsection{Multi-Recipient Key Encapsulation Mechanism} \label{sec:mKEM}
\begin{definition} [Multi-Recipient Key Encapsulation Mechanism] \label{def:mKEM}
A \emph{(single-message) multi-recipient key encapsulation mechanism ($\mKEM$)} over a key space $\KeySpace$ and ciphertext space $\CtSpace$ consists of the following five algorithms $\mKEM = (\mSetup, \mGen, \mEncaps, \allowbreak \mExt, \mDecaps):$ 
\begin{itemize}
	\item $\mSetup(1^\secpar) \rightarrow \pp:$ The setup algorithm on input the security parameter $1^\secpar$ outputs a public parameter~$\pp$. 

	\item $\mGen(\pp) \rightarrow (\pk, \sk):$ The key generation algorithm on input a public parameter $\pp$ outputs a pair of public key and secret key $(\pk, \sk)$. 
		
	\item $\mEncaps(\pp, \tuple{\pk_i}_{i \in [N]}) \rightarrow (\key, \vct):$ The encapsulation algorithm on input a public parameter $\pp$, and $N$
	                                                                                                              public keys $\tuple{\pk_i}_{i \in [N]}$, outputs a key $\key$ and a ciphertext $\vct$. 
	
	\item $ \mExt(i, \vct) \rightarrow \ct_i:$ The deterministic extraction algorithm on input an index $i \in \NN$ and a ciphertext $\vct$, outputs either $\ct_i$ or a special symbol $\bot_\Ext$ indicating extraction failure. 
	
	\item $\mDecaps( \sk, \ct_i ) \rightarrow \key \text{ or } \bot:$ The deterministic decryption algorithm on input a secret key $\sk$ and a ciphertext $\ct_i$, outputs either $\key \in \KeySpace$ or a special symbol $\bot \not \in \KeySpace$. 
\end{itemize}

\end{definition}

\begin{definition} [Correctness] \label{def:mKEM_correctness} 
A $\mKEM$ is $\delta_N$-correct if
\begin{equation*}
\delta_N \ge  \Pr\left[ \begin{array}{c} (\key, \vct) \lruns \mEnc(\pp, \tuple{\pk_i}_{i \in [N]}) ,  \tuple{ \ct_i \lruns \mExt(i, \vct)}_{i \in [N]}; 
 \\ \exists i \in [N] ~s.t.~ \key \neq \mDec( \sk, \ct_i ) \end{array} \right],
\end{equation*}
where the probability is also taken over $\pp \lruns \mSetup$ and $(\pk_i, \sk_i) \lruns \mGen(\pp)$ for all $i \in [N]$. 
\end{definition}

\begin{definition} [$\INDCCA$]
Let $\mKEM$ be a multi-recipient $\KEM$. We define $\INDCCA$ by a game illustrated in \ref{fig:CPACCAgame} and say the (possibly quantum) adversary $\A$ 
(making only classical decapsulation queries to $\OD$) \emph{wins} if the game outputs $1$. 
We define the advantage of $\A$ against $\INDCCA$ security of $\mKEM$ parameterized by $N \in \NN$ as
$$
\Adv^{\INDCCA}_{\mKEM, N} (\A)  = \left|  \Pr[ \A \text{ \emph{wins}} ] - 1/2 \right|. 
$$
\end{definition}


\subsubsection{Generic Construction via FO Transform} \label{sec:generic_INDCCA_KEM}
The authors of~\cite{MR_PKE2020} provided a generic transformation of an $\INDCPA$ secure $\mPKE$ to an $\INDCCA$ secure $\mKEM$ using a generalized Fujisaki-Okamoto transform, see Fig.~\ref{fig:FO}.

\begin{figure}[htb!]
\begin{minipage}[t]{0.30\textwidth}
\underline{ $\mSetup(1^\secpar) $}
    \begin{algorithmic}[1]
    \State $ \pp \lruns \mSetupp(1^\secpar) $
	\State \Return $\pp$
    \end{algorithmic}    
\end{minipage}
\hfill
\begin{minipage}[t]{0.35\textwidth}
\underline{ $\mGen( \pp )$ }
    \begin{algorithmic}[1]
    	\State $ (\pk, \skp) \lruns \mGenp(\pp) $
	\State $\seed \lruns \bin^\ell$
	\State $\sk := (\skp, \seed)$
	\State \Return $(\pk, \sk)$
    \end{algorithmic}
\end{minipage}
\hfill
\begin{minipage}[t]{0.30\textwidth}
\underline{ $\mExt( i, \vct )$ }
    \begin{algorithmic}[1]
    	\State $ \ct_i \lruns \mExtp(i, \vct) $
	\State \Return $\ct_i$
    \end{algorithmic}
\end{minipage}

\bigskip

\begin{minipage}[t]{0.5\textwidth}
\underline{ $\mEncaps( \pp, \tuple{\pk_i}_{i \in [N]} )$ }
    \begin{algorithmic}[1]
    	\State $\msg \samp \MsgSpace$
        \State $\ctnpk := \mEncA( \pp;  \OG_1(\msg) )$		
	\For{$i \in [N]$}
	   	\State {$\hct_i := \mEncB( \pp, \pk_i, \msg; \OG_1(\msg), \OG_2(\pk_i, \msg)  )$}
	\EndFor
        \State $\key := \OH(\msg)$
        \State \Return $(\key, \vct := (\ctnpk, \tuple{\hct_i}_{i \in [N]}))$
    \end{algorithmic}
\end{minipage}
\hfill
%
\begin{minipage}[t]{0.55\textwidth}
\underline{ $\mDecaps( \sk, \ct )$ }
    \begin{algorithmic}[1]
    	\State $ \sk := (\skp, \seed) $
    	\State $\msg := \mDec( \skp, \ct )$
	\If{ $\msg = \bot$}
		\State \Return $\key := \OH'(\seed, \ct) $
        	\EndIf
	\State $\ctnpk := \mEncA( \pp; \OG_1(\msg) )$
	\State  $\hct := \mEncB( \pp, \pk, \msg;\OG_1(\msg), \OG_2(\pk, \msg)  )$	
	\If{ $\ct \neq (\ctnpk, \hct)$ }
		\State \Return $\key := \OH'(\seed, \ct) $
	\Else 
	        \State \Return $\key := \OH( \msg )$
        	\EndIf
    \end{algorithmic}
\end{minipage}
%
\caption{ An $\INDCCA$ secure $\mKEM$ from a decomposable $\INDCPA$ secure $\mPKE = (\mSetupp, \mGenp, \mEnc = (\mEncA, \mEncB), \mExtp, \mDec)$. We include the superscript $^{\sf p}$ to make the code more readable. 
}\label{fig:FO}
\end{figure}


\begin{theorem} [$\INDCPA$ $\mPKE$ $\Rightarrow$ $\INDCCA$ $\mKEM$, adapted from~\cite{MR_PKE2020}] \label{thm:classical_proof_FO}
Assume $\mPKE$ with message space $\MsgSpace$ is $\delta$-correct and $\gamma$-spread.
Then, for any classical PPT $\INDCCA$ adversary $\A$ issuing at most $\qD$ queries to the decapsulation oracle $\OD$, a total of at most $\qG$ queries to $\OG_1$ and $\OG_2$, and at most $\qH$ and $\qH'$ queries to $\OH$ and $\OH'$, there exists a classical PPT adversary $\BIND$ such that 
\begin{align*}
	\Adv^{\INDCCA}_{\mKEM, N}( \A )  &\le  2 \cdot \Adv^{\INDCPA}_{\mPKE, N}(\BIND) + (2\qG + \qD + 2)\cdot\delta \\
    &+ \qD \cdot 2^{-\gamma} + \frac{(\qG + \qH)}{|\MsgSpace|}  + \qH' \cdot N \cdot 2^{-\ell}. 
\end{align*}
where the running time of $\BIND$ is about that of $\A$, and $\ell$ is bit-length of the seed included in the private key. 
\end{theorem}

\begin{theorem} [$\INDCPA$ $\mPKE$ $\Rightarrow$ $\INDCCA$ $\mKEM$, adapted from~\cite{MR_PKE2020}] \label{thm:quantum_proof_FO}
Assume $\mPKE$ with message space $\MsgSpace$ is $\delta$-correct and $\gamma$-spread.
Then, for any quantum PT $\INDCCA$ adversary $\A$ issuing at most $\qD$ classical queries to the decapsulation oracle $\OD$, a total of at most $\qG$ quantum queries to $\OG_1$ and $\OG_2$, and at most $\qH$ and $\qH'$ quantum queries to $\OH$ and $\OH'$, there exists a quantum PT adversary $\BIND$ such that 
\begin{align*}
	\Adv^{\INDCCA}_{\mKEM, N}( \A ) &\le   \sqrt{8\cdot  (\qG + 1)\cdot \Adv^{\INDCPA}_{\mPKE, N}(\BIND) }+ \frac{8\cdot (\qG + 1)}{\sqrt{|\MsgSpace|} }  \\
	 &\quad + 12\cdot(\qG + \qD + 1)^2 \cdot \delta_N + \qD\cdot 9\sqrt{2^{-\gamma}} + 9 \cdot 2^{\frac{-\mu}{2}}+ \qH' \cdot N \cdot 2^{\frac{-\ell+1}{2}}, 
\end{align*}
where the running time of  $\BIND$ is about that of $\A$, $\ell$ is bit-length of the seed included in the private key, and $\mu = \abs{\randnpk} + \abs{\rand}$ for $(\randnpk, \rand) \in \RandSpace$ where $\RandSpace$ is the randomness space of $\mPKE$ for a single ciphertext. 
\end{theorem}

\end{document}